\documentclass{CSML}

\def\dOi{12(4:10)2016}
\lmcsheading%
{\dOi}
{1--34}
{}
{}
{Dec.~11, 2015}
{Dec.~28, 2016}
{}

\ACMCCS{[{\bf Theory of computation}]:  Semantics and
  reasoning---Program semantics---Categorical semantics; Logic---Modal
  and temporal logics} 

\usepackage{amsmath,amssymb}
\usepackage[all]{xy}
\usepackage{stmaryrd}
\usepackage{wrapfig}
\usepackage{graphicx}
\usepackage{hyperref}

\theoremstyle{plain}
\newtheorem{lemma}[thm]{Lemma}

\newtheorem{theorem}[thm]{Theorem}
\newtheorem{corollary}[thm]{Corollary}
\theoremstyle{definition}
\newtheorem{definition}[thm]{Definition}
\newtheorem{example}[thm]{Example}
\newtheorem{remark}[thm]{Remark}

\newcommand{\op}{\mathit{op}}
\newcommand{\coalg}[1]{\mathsf{Coalg}(#1)}
\newcommand{\alg}[1]{\mathsf{Alg}(#1)}
\newcommand{\Set}{\mathsf{Set}}
\newcommand{\PCM}{\mathsf{PCM}}
\newcommand{\M}{\mathcal{M}}
\newcommand{\C}{\mathcal{C}}
\newcommand{\D}{\mathcal{D}}
\newcommand{\sem}{\mathbb{S}}
\newcommand{\tr}{{s^\flat}}
\newcommand{\trr}[1]{{s^\flat_{#1}}}
\newcommand{\Id}{\mathsf{Id}}
\newcommand{\id}{\mathsf{id}}
\newcommand{\pow}{\mathcal{P}}
\newcommand{\powf}{{\pow_\omega}}
\newcommand{\true}{\mathsf{tt}}

\newcommand{\To}{\Rightarrow}

\newcommand{\coalgfun}[1]{#1 \circ -}
\newcommand{\algfun}[1]{- \circ #1}

\begin{document}

\title[Coalgebraic trace semantics via forgetful logics]{Coalgebraic trace semantics via forgetful logics}

\author[B.~Klin]{Bartek Klin\rsuper a}	
\address{{\lsuper a}University of Warsaw}	
\email{klin@mimuw.edu.pl}  
\thanks{{\lsuper a}Supported by the Polish National Science Centre (NCN) grant 2012/07/E/ST6/03026}	

\author[J.~Rot]{Jurriaan Rot\rsuper b}	
\address{{\lsuper b}Radboud University, Nijmegen}	
\email{jrot@cs.ru.nl}  
\thanks{{\lsuper b}The research leading to these results has received funding from the
European Research Council under the European Union's Seventh Framework
Programme (FP7/2007-2013) / ERC grant agreement nr.~320571.}



\begin{abstract}
We use modal logic as a framework for coalgebraic trace semantics, and show the 
flexibility of the approach with concrete examples such as the language semantics
of weighted, alternating and tree automata, and the trace semantics of
generative probabilistic systems.
We provide a sufficient condition under which a logical semantics coincides with 
the trace semantics obtained via a given determinization construction.
Finally, we consider a condition that guarantees the
existence of a canonical determinization procedure that
is correct with respect to a given logical semantics.
That procedure is closely related to  
Brzozowski's minimization algorithm.
\end{abstract}

\maketitle

\section{Introduction}

The theory of coalgebras~\cite{Rutten00,Jacobs:coalg} is a framework of choice to model and study state-based systems at a high level of generality. Coalgebraic methods have been rather successful in modeling branching time behaviour of various kinds of transition systems, with a general notion of bisimulation and final semantics as the main contributions. 
But in a wide range of settings, from automata theory to verification and model checking, one is interested
in the linear time behaviour of systems, such as trace semantics of transition systems or language semantics of automata. 
Indeed, coalgebraic modeling of linear time behaviour has also attracted significant attention. 

However, the emerging picture is considerably more complex: several approaches have been developed whose scopes and connections are not yet fully understood. Here we suggest a new approach which, contrary to previously considered coalgebraic ideas,
crucially and explicitly defines trace semantics not by coinduction but by induction, based on techniques
from modal logic. This provides a transparent framework that covers new examples such as weighted tree automata, allows us to naturally
relate trace semantics to determinization constructions, and, in certain well-behaved cases,
provides canonical determinization and minimization procedures.

To study trace semantics coalgebraically, one usually considers systems whose behaviour type is a composite functor of the form $TB$ or $BT$, where $T$ represents a branching aspect of behaviour that trace semantics is supposed to ``resolve'', and $B$ represents the transition aspect that should be recorded in system traces. Typically it is assumed that $T$ is a monad, and its multiplication structure is used to resolve branching. For example, in~\cite{powerturi,HJS07}, a distributive law of $B$ over $T$ is used to lift $B$ to the Kleisli category of $T$, and trace semantics is obtained as final semantics for the lifted functor. Additional assumptions on $T$ are needed for this, so this approach does not work for coalgebras such as weighted automata. 
On the other hand, in~\cite{JSS14,SBBR13} a distributive law of $T$ over $B$ is used to lift $B$ to the Eilenberg-Moore category of $T$, with trace semantics again obtained as final semantics for the lifted functor. This can be seen as a coalgebraic generalization of the powerset determinization procedure for non-deterministic automata.
While it applies to many examples, that approach does not work for systems that do not determinize, such as tree automata.
A detailed comparison of these two approaches is in~\cite{JSS14}. 

In this paper, we study trace semantics in terms of modal logic. The basic idea is very simple: we view traces as formulas in suitable modal logics, and trace semantics of a state arises from all formulas that hold for it. 
A coalgebraic approach to modal logic based on dual adjunctions is by now well developed~\cite{pavlovic,klin07,jacobssokolova09,kupkepattinson11}, and we apply it to speak of traces generally. Obviously not every logic counts as a trace logic: assuming a behaviour type of the form $BT$ or $TB$, we construct logics from arbitrary (but usually expressive) logics for $B$ and special logics for $T$ whose purpose is to resolve branching. We call such logics {\em forgetful}.

Our approach differs from previous studies in a few ways:
\begin{itemize}
\item Trace semantics is obtained not as final semantics of coalgebras, but by initial semantics of algebras.
Fundamentally, we view trace semantics as an inductive concept and not a coinductive one akin to bisimulation, 
although in some well-behaved cases the inductive and coinductive views coincide.
\item We do not assume that $T$ is a monad, unless we want to relate our logical approach to ones that do, in particular to determinization constructions. 
\item Instead of using monad multiplication $\mu\colon TT\To T$ to resolve branching, we use a natural transformation $\alpha\colon TG\To G$, where $G$ is a contravariant functor that provides the basic infrastructure of logics. In case of nondeterministic systems, $T$ is the covariant powerset functor and $G$ the contravariant powerset, so $TT$ and $TG$ act the same on objects, but they carry significantly different intuitions.
\item Thanks to the flexibility of modal logics, we are able to cover examples such as the language semantics of 
weighted tree automata, that does not quite fit into previously studied approaches, or alternating automata.
\end{itemize}

\begin{example}
Consider weighted automata over a semiring $\sem$. Every state of such an automaton, given an input letter,
has a linear combination of successor states, with coefficients in $\sem$. The
semantics of a weighted automaton with state space $X$ is a function $\tr \colon X \rightarrow \sem^{A^*}$,
mapping every state to a weighted language. In our framework, the 
actual map $\tr$ is computed as the transpose of a certain map $s \colon A^* \rightarrow \sem^X$,
which arises by induction on words. This computation is determined by a certain (rather simplistic) {\em forgetful logic}; the
formulas of this logic are just the words in $A^*$. 
In this example, it is crucial to distinguish between the branching structure of automata, modeled by linear combinations of states, and the functor $\sem^{-}$ that allows arbitrary assignments of weights to words and states. This distinction is an instance of a general distinction between a covariant functor (or a monad) $T$ and a contravariant functor $G$, 
which will be a recurring theme in this paper.
\end{example}

Applying the term {\em logic} to our framework is admittedly a bit of an exaggeration. We do not consider syntactic aspects usually connected with logical systems, such as deduction relations. We merely study ways to interpret certain terms (formulas) on models (coalgebras) in an inductive fashion. One might call it a testing framework, as in~\cite{pavlovic}, but similar systems are often called logics in coalgebraic literature, and we reluctantly stick to this custom.

The idea of using modal logics for coalgebraic trace semantics is not new; it is visible already in~\cite{pavlovic}.
In~\cite{HJS07} it is related to behavioural equivalence, and applied to non-deterministic systems.
A generalized notion of relation lifting is used in~\cite{Cirstea13} to obtain infinite trace semantics,
and applied in~\cite{Cirstea14} to get canonical linear time logics.
In~\cite{kissigKurz}, coalgebraic modal logic is combined with the idea of lifting behaviours to Eilenberg-Moore categories, 
with trace semantics in mind. In~\cite{JSS14}, a connection to modal logics is sketched from the perspective of coalgebraic determinization procedures. In a sense, this paper describes the same connection from the perspective of logic. In~\cite{corina15}, the framework of~\cite{Cirstea13,Cirstea14} is rephrased in terms of coalgebraic modal logic. The result is rather similar to the one we considered, with forgetful modalities to resolve branching. Unlike our approach,~\cite{corina15} relies on $T$ being a monad, and under some more assumptions it studies canonical forgetful modalities that give rise to particularly well-behaved logics. In~\cite{KMPS15,lutz15}, monads feature even more prominently, with the entire behaviour functor embedded in a so-called graded monad.
In~\cite{Goncharov13}, it is embedded in a more complex functor with a so-called observer.

Our main new contribution is the notion of forgetful logic and its ramifications. Basic definitions are provided in Section~\ref{sec:forgetful} and some illustrative examples in Section~\ref{sec:examples}.
We introduce a systematic way of relating trace semantics to determinization, by giving sufficient conditions 
for a given determinization procedure, understood in a slightly more general way than in~\cite{JSS14}, 
to be correct with respect to a given forgetful logic (Section~\ref{sec:determinization}). For instance, this allows showing in a coalgebraic setting that the determinization of alternating automata into non-deterministic automata preserves language semantics.

A correct determinization procedure may not exist in general. In Section~\ref{sec:isomates} we study a situation where a canonical correct determinization procedure exists. It turns out that even in the simple case of non-deterministic automata that procedure is not the classical powerset construction; instead, it relies on a double application of contravariant powerset construction. Interestingly, this is what also happens in Brzozowski's algorithm for automata minimization~\cite{brzozowski}, so as a by-product, we get a new perspective on that algorithm which has recently attracted much attention in the coalgebraic community~\cite{AdamekBHKMS12,BezhanishviliKP12,BBHPRS14}.

Although we do not assume the branching functor $T$ to be a monad, a forgetful logic for $T$ is equivalent to a transformation from $T$ to a certain monad which, in the case of sets, is the double contravariant powerset monad (a special case of the continuation monad). One might say that the continuation monad is rich enough to handle all types of branching that can be ``forgotten'' within our framework.

This paper is an extended version of~\cite{KlinR15}, adding full proofs and a treatment
of the trace semantics of probabilistic systems as a non-trivial instance of the framework.  

\emph{Acknowledgments} We thank Marcello Bonsangue, Helle Hvid Hansen, Ichiro Hasuo, Bart Jacobs and Jan Rutten for discussions. Joost Winter spotted a serious mistake in a previous version of Example~\ref{ex:alt-correct}. We are very grateful to the anonymous reviewer who pointed out Lemma~\ref{lem:monadmap} and whose many insightful comments let us significantly improve the paper.

\section{Preliminaries}\label{sec:prelim}

We assume familiarity with basic notions of category theory (see, e.g.,~\cite{mac1998categories}).
A coalgebra for a functor $B \colon \C \rightarrow \C$ consists of an object $X$ and a map $f \colon X \rightarrow BX$. 
A homomorphism from $f \colon X \rightarrow BX$ to $g \colon Y \rightarrow BY$ is a map $h \colon X \rightarrow Y$
such that $g \circ h = Bh \circ f$. The category of $B$-coalgebras is denoted by $\coalg{B}$.
Algebras for a functor $L$ are defined dually; the category of $L$-algebras and homomorphisms is denoted by $\alg{L}$.

We list a few examples where $\C = \Set$, the category of sets and functions.
  Consider the functor $\powf(A \times -)$, where $\powf$ is the finite powerset functor and $A$ is a fixed set. A coalgebra $f \colon X \rightarrow \powf(A \times X)$ is a finitely
branching labelled transition system: it maps every state to a finite set of next states. Coalgebras
for the functor $(\powf -)^A$ are image-finite labelled transition systems, i.e., the set of next states
for every label is finite. When $A$ is finite the two notions coincide.
A coalgebra  $f \colon X \rightarrow \powf(A \times X + 1)$, where $1= \{*\}$ is a singleton, is a non-deterministic automaton;
a state $x$ is accepting whenever $* \in f(x)$.

Consider the functor $BX = 2 \times X^A$, where 2 is a two-element set of truth values. 
A coalgebra $\langle o, f \rangle \colon X \rightarrow BX$ is a deterministic automaton; a state $x$
is accepting if $o(x)=\true$, and $f(x)$ is the transition function. The composition $B \powf$
yields non-deterministic automata, presented in a different way than above. We shall also consider
$B \powf \powf$-coalgebras, which represent a general version of alternating automata.

Let $\sem$ be a semiring.
Define 
$
  \M X = \{\varphi \in \sem^X \mid \text{supp}(\varphi) \text{ is finite}\}
$ where $\text{supp}(\varphi) = \{x \mid \varphi(x) \neq 0\}$, and  $
 \M(f \colon X \rightarrow Y)(\varphi)(y) = \sum_{x \in f^{-1}(y)}\varphi(x)
$.
A \emph{weighted automaton} is a coalgebra for the functor $\M (A \times - + 1)$.
Let $\Sigma$ be a polynomial functor corresponding to an algebraic signature.
A \emph{top-down weighted tree automaton} is a coalgebra for the functor $\M \Sigma$. 
For $\sem$ the Boolean semiring these are \emph{non-deterministic tree automata}. 
Similar to non-deterministic automata above,
one can present weighted automata as coalgebras for $\sem \times (\M-)^A$.

We note that $\powf$ is a monad, by taking the unit $\eta_X(x) = \{x\}$ and the multiplication $\mu$ to be set union.
More generally, the functor $\M$ extends to a monad,
by taking $\mu_X(\varphi)(x) = \sum_{\psi \in \sem^X} \varphi(\psi) \cdot \psi(x)$. The
case of $\powf$ is obtained by taking the Boolean semiring. 
Notice that the finite support condition is required for $\mu$ to be well-defined.
 
\subsection{Contravariant adjunctions}\label{sec:contr-adj}

The basic framework of coalgebraic logic is formed of two categories $\C$, $\D$ connected by functors $F\colon\C^{\op}\to \D$ and $G\colon\D^{\op}\to \C$ that form an adjunction $F^{\op}\dashv G$. For example, one may take $\C=\D=\Set$ and $F=G=2^-$, for $2$ a two-element set of logical values. The intuition is that objects of $\C$ are collections of processes, or states, and objects of $\D$ are logical theories.

To avoid cluttering the presentation with too much of the  $(-)^{\op}$ notation, we opt to treat $F$ and $G$ as {\em contravariant functors}, i.e., ones that reverse the direction of all arrows (maps), between $\C$ and $\D$. The adjunction then becomes a contravariant adjunction ``on the right'', meaning that there is a natural bijection
\[
	\C(X,G\Phi) \cong \D(\Phi,FX) \qquad \mbox{ for } X\in\C, \Phi\in\D.	
\] 
Slightly abusing the notation, we shall denote both sides of this bijection by $(-)^{\flat}$. Applying the bijection to a map is referred to as transposing the map.

In such an adjunction, $GF$ is a monad on $\C$, whose unit we denote by $\iota\colon\Id\To GF$, and $FG$ is a monad on $\D$, with unit the denoted by $\epsilon\colon\Id\To FG$.
As usual, multiplication of the monad $GF$ is $G\epsilon F \colon GFGF\To GF$, and multiplication of $FG$ is $F\iota G \colon FGFG\To FG$. The counit-unit equations for adjunctions amount to:
\begin{equation}\label{eq:counit-unit}
\vcenter{\xymatrix{
F\ar@{=>}[r]^-{\epsilon F}\ar@{=}[rd] & FGF\ar@{=>}[d]^{F\iota} \\
& F
}}
\qquad\qquad
\vcenter{\xymatrix{
G\ar@{=>}[r]^-{\iota G}\ar@{=}[rd] & GFG\ar@{=>}[d]^{G\epsilon} \\
& G
}}
\end{equation}
Both $F$ and $G$ map colimits to limits, by standard preservation results for adjoint functors.

In what follows, the reader need only remember that $F$ and $G$ are contravariant, i.e., they reverse maps and natural transformations. 
All other functors, except a few that lift $F$ and $G$ to other categories, are standard covariant functors.

\subsection{Coalgebraic modal logic}

We recall an approach to coalgebraic modal logic based on contravariant adjunctions, see, e.g.,~\cite{klin07,jacobssokolova09}. 
Consider categories $\C$, $\D$ and functors $F$, $G$ as in Section~\ref{sec:contr-adj}. 
Given an endofunctor $B\colon\C\to\C$, a {\em coalgebraic logic} to be interpreted on $B$-coalgebras is built of {\em syntax}, i.e., an endofunctor $L\colon\D\to\D$, and {\em semantics}, a natural transformation $\rho\colon LF\To FB$. We will usually refer to $\rho$ simply as a logic.
If an initial $L$-algebra $a\colon L\Phi\to\Phi$ exists then, for any $B$-coalgebra $h\colon X\to BX$, the {\em logical semantics} 
of $\rho$ on $h$ is a map $s^{\flat}\colon X\to G\Phi$ obtained by transposing the map defined by initiality of $a$.
\begin{equation}\label{eq:logsem1}
\vcenter{\xymatrix{
	L\Phi\ar[dd]_a\ar[r]^{Ls} & LFX\ar[d]^{\rho_X}\\
	& FBX\ar[d]^{Fh} \\
	\Phi\ar[r]_s & FX
}}
\end{equation}
The mapping of a $B$-coalgebra $h\colon X\to BX$ to an $L$-algebra $Fh\circ \rho_X\colon LFX\to FX$ determines a contravariant functor $\hat F$ that lifts $F$, i.e., acts as $F$ on carriers.
\begin{equation}\label{eq:liftF}
\vcenter{\xymatrix{
\coalg{B} \ar[d] \ar[r]^-{\hat F} & \alg{L} \ar[d] \\
\C \ar[r]_F & \D
}}
\end{equation}
The functor $\hat F$ has no (contravariant) adjoint in general; later in Section~\ref{sec:isomates} we shall study well-behaved situations when it does.
Notice that $\hat F$ maps coalgebra homomorphisms to algebra homomorphisms, and indeed the logical semantics 
factors through coalgebra homomorphisms, i.e., behavioural equivalence implies logical equivalence. 
The converse holds if $\rho$ is \emph{expressive}, meaning that the logical semantics decomposes as a coalgebra
homomorphism followed by a mono.


\begin{example}\label{ex:logic-da}
	Let $\C = \D = \Set$, $F = G = 2^-$, $B = 2 \times -^A$ and $L = A \times -+1$.
	The initial algebra of $L$ is the set $A^*$ of words over $A$. 
	We define a logic $\rho \colon LF \Rightarrow FB$ as follows:
	$
		\rho_X(*)(o,t) = o$ and $\rho_X(a,\varphi)(o,t) = \varphi(t(a)) 
	$.
	For a coalgebra $\langle o, f \rangle \colon X \rightarrow 2 \times X^A$
	the logical semantics is a map $\tr \colon X \rightarrow 2^{A^*}$, yielding the usual language semantics of the automaton: $\tr(x)(\varepsilon) = o(x)$ 
	for the empty word $\varepsilon$, and
	$\tr(x)(aw) = \tr(f(x)(a))(w)$ for any $a \in A, w \in A^*$.
\end{example}

Note that logical equivalences, understood as kernel relations of logical semantics, are conceptually different from behavioural equivalences typically considered in coalgebra theory, in that they do not arise from finality of coalgebras, but rather from initiality of algebras (albeit in a different category). Fundamentally, logical semantics for coalgebras is defined by induction rather than coinduction. In some particularly well-behaved cases the inductive and coinductive views coincide; we shall study such situations in Section~\ref{sec:isomates}.

A logic $\rho\colon LF\To FB$ gives rise to its {\em mate} $\rho^\flat\colon BG\To GL$, defined by
\begin{equation}\label{eq:mate}
\xymatrix{
	BG\ar@{=>}[r]^-{\iota BG}  & GFBG\ar@{=>}[r]^{G\rho G} & GLFG\ar@{=>}[r]^-{GL\epsilon} & GL,
}
\end{equation}
where $\iota$ and $\epsilon$ are as in Section~\ref{sec:contr-adj}. A routine calculation shows that $\rho$ in turn is the mate of $\rho^\flat$ (with the roles of $F$, $G$, $\iota$ and $\epsilon$ swapped), giving a bijective correspondence between logics and their mates. Some important properties of logics are conveniently stated in terms of their mates; e.g., 
under mild additional assumptions (see~\cite{klin07}), if the mate is pointwise monic then the logic is expressive.
Two simple but useful diagrams show how logics relate to their mates along the basic adjunction:
\begin{lemma}\label{lem:useful}
For any logic $\rho \colon LF\To FB$, the following diagrams commute:
\[\xymatrix{
	B\ar@{=>}[r]^-{\iota B}\ar@{=>}[d]_{B\iota} & GFB\ar@{=>}[d]^{G\rho} \\
	BGF\ar@{=>}[r]_-{\rho^{\flat}F} & GLF
}
\qquad\qquad
\xymatrix{
	L\ar@{=>}[r]^-{\epsilon L}\ar@{=>}[d]_{L\epsilon} & FGL\ar@{=>}[d]^{F\rho^{\flat}} \\
	LFG\ar@{=>}[r]_-{\rho G} & FBG.
}
\]
\end{lemma}
\begin{proof}
For the first diagram, chase:
\[\xymatrix{
	B\ar@{=>}[rrr]^{\iota B}\ar@{=>}[d]_{B\iota} & & & GFB\ar@{=>}[lld]_{GFB\iota}\ar@{=>}[d]^{G\rho} \\
	BGF\ar@{=>}[r]^-{\iota BGF}\ar@{=>}@/_1pc/[rrrd]_{\rho^{\flat}F} & GFBGF\ar@{=>}[r]_-{G\rho GF} & GLFGF\ar@{=>}[rd]_-{GL\epsilon F} & GLF\ar@{=}[d]\ar@{=>}[l]_-{GLF\iota} \\
	& & & GLF
}\]
where everything commutes, clockwise starting from top-left: by naturality of $\iota$, by naturality of $\rho$, by~\eqref{eq:counit-unit} above, and by definition of $\rho^{\flat}$. 

The other diagram is similar. 
\end{proof}


\noindent There is a direct characterization of logical semantic maps in terms of mates, first formulated in~\cite{pavlovic}. Indeed,
by transposing~\eqref{eq:logsem1} it is easy to check that 
the logical semantics $s^{\flat}\colon X\to G\Phi$ on a coalgebra $h \colon X\to BX$ is a unique map that makes the following ``twisted coalgebra morphism'' diagram commute.
\begin{equation}\label{eq:logsem2}
\vcenter{\xymatrix{
	BX\ar[r]^{Bs^{\flat}} & BG\Phi\ar[d]^{\rho^{\flat}_{\Phi}} \\
	& GL\Phi \\
	X\ar[uu]^h\ar[r]_{s^{\flat}} & G\Phi\ar[u]_{Ga}.
}}
\end{equation}

\subsection{Morphisms of logics} 
Given logics $\rho \colon LF \Rightarrow FB$ and $\theta \colon MF \Rightarrow FK$, a \emph{morphism} 
from $\theta$ to $\rho$ consists of a pair of natural transformations $\tau \colon M \Rightarrow L$
and $\kappa \colon B \Rightarrow K$ such that the following diagram commutes:
\begin{equation}
\vcenter{
\xymatrix{
	MF \ar@{=>}[r]^{\theta}  \ar@{=>}[d]_{\tau F}
		& FK \ar@{=>}[d]^{F\kappa} \\
	LF \ar@{=>}[r]_{\rho} 
		& FB 
}
}
\qquad \text{or, equivalently, } \qquad
\vcenter{
\xymatrix{
	BG \ar@{=>}[r]^{\rho^\flat}  \ar@{=>}[d]_{\kappa G}
		& GL \ar@{=>}[d]^{G\tau} \\
	KG \ar@{=>}[r]_{\theta^\flat} 
		& GM.
}}
\end{equation}
Natural transformations $\tau$ and $\kappa$ as above induce functors
$$
\xymatrix{
	\alg{L} \ar[r]^{\algfun{\tau}} & \alg{M}
}
\qquad \text{ and } \qquad
\xymatrix{
	\coalg{B} \ar[r]^{\coalgfun{\kappa}} & \coalg{K}
}
$$
defined by composition. 
\begin{lemma}\label{lm:logic-morph}
	Suppose $(\tau,\kappa)$ is morphism from $\theta \colon MF \Rightarrow FK$ to $\rho \colon LF \Rightarrow FB$. Then
	the following diagram commutes:
	$$
	\xymatrix{
		\coalg{B} \ar[r]^{\coalgfun{\kappa}} \ar[d]_{\bar{F}_\rho}
		& \coalg{K} \ar[d]^{\bar{F}_\theta} \\
		\alg{L} \ar[r]_{\algfun{\tau}} & \alg{M}
	}
	$$
	where $\bar{F}_\rho$ and $\bar{F}_\theta$ are the liftings of $F$ induced by $\rho$ and $\theta$ respectively, as in~\eqref{eq:liftF}.
\end{lemma}
\begin{proof}
	Let $f \colon X \rightarrow BX$ be a coalgebra, and consider:
	$$
	\xymatrix{
		MFX \ar[r]^{\theta_X} \ar@{=}[d]
			& FKX \ar[r]^{F\kappa_X} 
			& FBX \ar[r]^{Ff} \ar@{=}[d]
			& FX  \ar@{=}[d] \\
		MFX \ar[r]_{\tau_{FX}} 
			& LFX \ar[r]_{\rho_X} 
			& FBX \ar[r]_{Ff} 
			& FX
	}
	$$
	The upper path is $\bar{F}_\theta (\kappa_X \circ f)$, and the lower path is $(\bar{F}_\rho(f)) \circ \tau_X$. The diagram commutes
	since $(\tau, \kappa)$ is a morphism of logics.
\end{proof}
Morphisms of logics apppear in~\cite{ChenJ14}, where the category of logics and morphisms between them is studied. 
The examples in~\cite{ChenJ14} involve a translation of syntax determined by $\tau$. 
Our main interest in morphisms of logics is cases where $\tau = \id$. Then, it is a direct consequence of Lemma~\ref{lm:logic-morph}
that the logical semantics of $\rho$ on a coalgebra $f \colon X \rightarrow BX$ coincides with the logical semantics
of $\theta$ on $\kappa_X \circ f$. In the sequel, in a situation where $\tau=\id$ we simply say that $\kappa$ is a morphism of logics.

\section{Forgetful logics}\label{sec:forgetful}

In most abstract approaches to coalgebraic trace semantics, the behaviour functor under consideration is a composition $TB$ or $BT$, where $T$ is the branching aspect 
and $B$ is the type of observations of interest. 
Our approach is to capture trace semantics as the logical semantics of a suitable logic for $TB$ or $BT$. 
The logics that we consider are defined as the composition of a logic for $B$ and a special kind of logic for $T$ which has trivial syntax. 
This special logic for $T$ specifies how the branching behaviour should be ``forgotten'' in the resulting logical theory.

Logics for composite functors can often be obtained from logics of their components. Consider functors $B,T\colon\C\to\C$ and 
logics for them $\rho\colon LF\To FB$ and $\alpha\colon NF\To FT$, for some functors $L,N\colon\D\to\D$. One can then define logics for the functors $TB$ and $BT$:
\[
	\alpha\circledcirc\rho = \alpha B \circ N\rho\colon NLF\To FTB, \qquad \rho\circledcirc\alpha = \rho T \circ L\alpha\colon LNF \To FBT.
\]
It is easy to see that taking the mate of a logic respects this composition operator, i.e., that $(\alpha\circledcirc\rho)^\flat = \alpha^\flat\circledcirc\rho^\flat$. 
Such compositions of logics appear in~\cite{Jacobs:coalg} and were studied in a slightly more concrete setting in~\cite{cirsteaPattinson04,SP11}. In~\cite[Lemma 3.12]{ChenJ14}, it is shown how to (horizontally) compose
morphisms of logics, turning $\circledcirc$ into a bifunctor.

We shall be interested in the case where the logic for $T$ has a trivial syntax; in other words, where $N=\Id$. Intuitively speaking, we require a logic for $T$ that consists of a single unary operator, which could therefore be elided in a syntactic presentation of logical formulas. The semantics of such an operator is defined by a natural transformation $\alpha\colon F\To FT$ or equivalently by its mate $\alpha^\flat\colon TG\To G$. Intuitively, the composite logics $\alpha\circledcirc\rho$ and $\rho\circledcirc\alpha$, when interpreted on $TB$- and $BT$-coalgebras respectively disregard, or forget, the aspect of their behaviour related to the functor $T$, in a manner prescribed by $\alpha$. 
We call logics obtained in this fashion {\em forgetful logics}.

Transformations $\alpha\colon F\To FT$ (and their mates $\alpha^\flat$) are also in bijective correspondence with natural transformations $\alpha^\dagger \colon T\To GF$. Indeed, define
\[
	\alpha^\dagger = G\alpha \circ \iota T \qquad\mbox{or equivalently}\qquad \alpha^\dagger = \alpha^\flat F\circ T\iota,
\]
and recover
\[
	\alpha = F\alpha^\dagger\circ \epsilon F \qquad\mbox{and} \qquad \alpha^\flat = G\epsilon\circ\alpha^\dagger G
\]
in what is easily seen to be mutually inverse operations. This means that, intuitively, a forgetful logic for $T$ is equivalent to an {\em encoding} of $T$ into the monad $GF$. In all examples considered in the next section $\alpha^\dagger$ is (pointwise) monic, which justifies the name ``encoding'', but we do not use the monicity for anything and we do not have an understanding of its significance.

As long as we do not assume $T$ to be a monad, it makes no sense to ask e.g.~whether $\alpha^\dagger$ is a monad morphism. However, composing forgetful logics for multiple behaviour functors does agree with the multiplication structure of $GF$. Specifically, for $\alpha \colon F\To FT$ and $\beta \colon F\To FS$, one may consider the composite $\alpha\circledcirc\beta \colon F\To FTS$. The corresponding encoding $(\alpha\circledcirc\beta)^\dagger \colon TS\To GF$ is then equal to:
\[\xymatrix{
	TS \ar@{=>}[r]^-{\alpha^\dagger\beta^\dagger} & GFGF\ar@{=>}[r]^-{G\epsilon F} & GF.
}\]
Encodings $\alpha^\dagger$ will be technically useful in Section~\ref{sec:det-preproc}.

\section{Examples}\label{sec:examples}

We instantiate the setting of Section~\ref{sec:forgetful}
and use forgetful logics to obtain trace semantics for several concrete
types of coalgebras: non-deterministic automata, transition systems, alternating automata, weighted tree automata and
probabilistic systems.


In the first few examples we let $\C = \D = \Set$ and $F = G = 2^-$,
and consider $TB$ or $BT$-coalgebras, where $T = \powf$ is the finite powerset
functor. Our examples involve the logic $\alpha \colon 2^- \Rightarrow 2^\powf$ defined by:
\begin{equation}\label{eq:alpha}
	\alpha_X(\varphi)(S) = \true \iff \exists x \in S . \varphi(x) = \true. 
\end{equation}
This choice of $F$ and $G$ has been studied thoroughly in the field of coalgebraic logic, and our $\alpha$ is an example of the standard notion of predicate lifting~\cite{Jacobs:coalg,kupkepattinson11} corresponding to the so-called diamond modality.
Its mate $\alpha^\flat \colon \powf 2^- \Rightarrow 2^-$ and the corresponding encoding $\alpha^\dagger \colon \powf\To 2^{2^-}$ are as follows: 
\begin{align*}
	\alpha_\Phi^\flat(S)(w) = \true &\iff \exists \varphi \in S . \varphi(w) = \true \\
	\alpha_X^{\dagger}(U)(\varphi) = \true &\iff \exists x \in U . \varphi(x) = \true
\end{align*}
Here and in all examples below, $\powf$ could be replaced by the full powerset $\pow$ without any problems.

\begin{example}\label{ex:nda}
We define a forgetful logic $\alpha \circledcirc \rho$ for the functor $\powf (A \times - +1)$ whose coalgebras 
are non-deterministic automata, so that the logical semantics is the usual language semantics. To this end, we let:
\begin{itemize}
	\item $\C = \D = \Set$, $F=G=2^-$;
	\item $T = \powf$, $B=L=(A \times - +1)$;
	\item $\alpha$ be as in~\eqref{eq:alpha}, and
	\item $\rho$ be defined by its mate $\rho^\flat \colon A\times 2^-+1 \rightarrow 2^{A \times - + 1}$ as follows:
\begin{align*}
	&\rho^\flat_\Phi(*)(t) = \true \iff t = *  \qquad \rho^\flat_\Phi(a,\varphi)(t) = \true \iff t = (a,w) \text{ and } \varphi(w) = \true \,.
\end{align*}
for any set $\Phi$.
\end{itemize}
The choice of $L$ is motivated by the fact that the initial algebra of $A \times - + 1$ is $A^*$, hence the
logical semantics will be a map from states to languages (elements of $2^{A^*}$). Now the logical 
semantics of the logic $\alpha \circledcirc \rho$ on an automaton $f \colon X \rightarrow \powf B X$ 
is the map $\tr$ from~\eqref{eq:logsem2}, i.e., the unique map that makes the following diagram commute: 
$$
\xymatrix@C=1.2cm{
	X \ar[d]_f \ar[rrr]^{\tr} & & & 2^{A^*} \ar[d] \\
	\powf (A \times X + 1) \ar[r]_-{\powf B \tr} 
	&\powf (A \times 2^{A^*} + 1) \ar[r]_-{\powf \rho^\flat_{A^*}} 
		&\powf(2^{A \times A^* + 1}) \ar[r]_-{\alpha^\flat_{LA^*}}
		&2^{A \times A^* + 1}
}
$$
It is easy to calculate (see Appendix~\ref{sec:details-examples}) that for any $x \in X$:
\begin{align*}
	\tr(x)(\varepsilon) = \true & \iff *\in f(x), \\
	\tr(x)(aw) = \true & \iff \exists y \in X . (a,y) \in f(x) \text{ and } \tr(y)(w) = \true,
\end{align*}
for $\varepsilon$ the empty word, and
for all $a \in A$ and $w \in A^*$.
\end{example}

Note that the logic $\rho$ in the above example is expressive. One may expect that given a different expressive logic $\theta$
involving the same functors, the forgetful logics $\alpha \circledcirc \rho$ and $\alpha \circledcirc \theta$ yield the same logical 
equivalences, but this is not the case. For instance, define $\theta^\flat \colon BG \Rightarrow GL$ 
as $\theta_\Phi^\flat(*)(t) = \true$ for all $t$, and $\theta^\flat_\Phi(a,\varphi) = \rho^\flat_\Phi(a,\varphi)$. This logic is expressive as well
(since $\theta^\flat$ is componentwise monic) but in the semantics of the forgetful logic $\alpha \circledcirc \theta$, 
information on final states is discarded.

\begin{example}[Length of words]
The initial algebra of the functor $L$ defined by $LX = X + 1$ is $\Phi = \mathbb{N}$, the set of natural numbers. 
Define a logic for $BX = A \times X + 1$ by its mate $\rho^\flat \colon A \times 2^- + 1 \Rightarrow 2^{- + 1}$ as follows:
\begin{align*}
	\rho^\flat(*)(t) = \true \iff t = *  \qquad 
	\rho^\flat(a,\varphi)(t) = \true \iff t = x \text{ and } \varphi(x) = \true \,.
\end{align*}
Note that this logic is not expressive.
With the above $\alpha$ (Equation~\ref{eq:alpha}), we have a logic $\alpha \circledcirc \rho$,
and given any $f \colon X \rightarrow \powf(A \times X + 1)$, this yields the following map
$\tr \colon X \rightarrow 2^{\Phi}$:
\begin{align*}
	\tr(x)(0) = \true 
		& \iff * \in f(x) , \\
	\tr(x)(n+1) = \true 
		& \iff \exists a \in A, y \in X. (a,y) \in f(x) \text{ and } \tr(y)(n) = \true \, .
\end{align*}	
Thus, $\tr(x)$ is the binary sequence which is $\true$ at position $n$ iff
the automaton $f$ accepts a word of length $n$, starting in state $x$.
\end{example}



\begin{example}[Labelled transition systems]\label{ex:log:lts}
In this example we consider the finite traces of labelled transition systems of the form $f \colon X \rightarrow (\powf X)^A$, i.e., 
$BT$-coalgebras where $BX = X^A$ and $TX = \powf X$. To this end, let $LX = A \times X + 1$. 
Define $\rho^\flat \colon (2^-)^A \Rightarrow 2^{A \times - + 1}$
as follows:
\begin{equation}\label{eq:ts-logic}
	\rho_\Phi^\flat(\varphi)(*) = \true  \qquad \rho_\Phi^\flat(\varphi)(a,w)=\varphi(a)(w) \, .
\end{equation}
Then the logical semantics $\tr \colon X \rightarrow 2^{A^*}$ of $\rho \circledcirc \alpha$ on a 
transition system $f \colon X \rightarrow (\powf X)^A$ is given by $\tr(x)(\varepsilon) = \true$
and $\tr(x)(aw)= \true$ iff $\tr(y)(w)=\true$ for some $y \in f(x)(a)$.
\end{example}

\begin{example}[Non-deterministic automata as $BT$-coalgebras]\label{ex:nda-bt}
Consider the functor $BX = 2 \times X^A$. 
Let $LX = A \times X + 1$, let $\rho^\flat \colon 2 \times (2^-)^A \Rightarrow 2^{A \times - + 1}$
be the mate of the logic $\rho$ given in Example~\ref{ex:logic-da}; explicitly, it is
the obvious isomorphism given by manipulating exponents:
\begin{equation}\label{eq:aut-logic}
	\rho_\Phi^\flat(o, \varphi)(*) = o  \qquad \rho_\Phi^\flat(o,\varphi)(a,w)=\varphi(a)(w)
\end{equation}
The logical semantics $\tr \colon X \rightarrow 2^{A^*}$ of $\rho \circledcirc \alpha$ on a 
coalgebra $\langle o, f \rangle \colon X \rightarrow 2 \times (\powf X)^A$
is the usual language semantics, i.e., for any $x \in X$ we have:
\begin{align*}
  \tr(x)(\varepsilon) =o(x) \qquad  
  \tr(x)(aw) = \true \iff \tr(y)(w)=\true \text{ for some }y \in f(x)(a) \, .
\end{align*}

Non-determinism can be resolved differently: in contrast to~\eqref{eq:alpha}, consider $\beta \colon \powf\To \powf 2^-$, and the corresponding $\beta^{\flat} \colon \powf 2^-\To\powf$ and $\beta^{\dagger} \colon \powf\To 2^{2^-}$, given by 
\begin{align*}
	\beta_X(\varphi)(S) = \true &\iff \forall x \in S . \varphi(x) = \true \\
	\beta_\Phi^\flat(S)(w) = \true &\iff \forall \varphi \in S . \varphi(w) = \true \\
	\beta_X^{\dagger}(U)(\varphi) = \true &\iff \forall x \in U . \varphi(x) = \true
\end{align*}
Similarly to~\eqref{eq:alpha}, $\beta$ is a predicate lifting that corresponds to the so-called box modality.
The semantics $\tr$ induced by the forgetful logic $\rho \circledcirc \beta$ accepts a word 
if \emph{all} paths end in an accepting state: $\tr(x)(\varepsilon)=o(x)$, and $\tr(x)(aw)= \true$ iff $\tr(y)(w)=\true$ for all $y \in f(x)(a)$. 
We call this the conjunctive semantics.
In automata-theoretic terms, this is the language semantics for ($B\powf$-coalgebras understood as) 
co-nondeterministic automata, i.e., alternating automata with only universal states.
\end{example}


\noindent
{\em Some non-examples.} It is not clear how to use forgetful logics to give a conjunctive semantics 
to coalgebras for $\powf(A \times X + -)$; simply using $\beta$ together with $\rho$ from Example~\ref{ex:nda}
does not yield the expected logical semantics.
Also, transition systems as $\powf (A \times -)$-coalgebras do not work well; with $\alpha$
as in~\eqref{eq:alpha} the logical semantics of a state with no successors is always empty, 
while it should contain the empty trace.


\begin{example}[Alternating automata]\label{ex:alt-aut}
Consider $B \powf \powf$-coalgebras with $B = 2 \times -^A$. We give a forgetful logic by combining
$\rho$, $\alpha$, and $\beta$ from the previous example 
(more precisely, the logic is $(\rho \circledcirc \alpha) \circledcirc \beta$); recall that $\alpha$ 
and $\beta$ resolve the non-determinism by disjunction
and conjunction respectively.
The logical semantics on a coalgebra $\langle o, f \rangle \colon X \rightarrow B \powf \powf X$ 
then is the map $\tr$ in the following diagram (see~\eqref{eq:logsem2}):
\begin{equation}\label{eq:aa-trace}
\vcenter{
\xymatrix@C=1.2cm{
	X \ar[d]_{\langle o, f \rangle} \ar[rrrr]^{\tr} & & & & 2^{A^*} \ar[d] \\
	B \powf \powf X \ar[r]_-{B \powf \powf \tr} 
	&B \powf \powf 2^{A^*} \ar[r]_-{B \powf \beta^\flat_{A^*}}
	&B \powf 2^{A^*} \ar[r]_-{B \alpha^\flat_{A^*}} 
	&B 2^{A^*} \ar[r]_-{\rho^\flat_{A^*}}  
		&2^{LA^*}
}
}
\end{equation}
Spelling out the details for a coalgebra $\langle o, f \rangle \colon X \rightarrow 2 \times (\powf \powf X)^A$ yields, for any $x \in X$: $\tr(x)(\varepsilon)=o(x)$ and for any $a \in A$ and $w \in A^*$:
$\tr(x)(aw) = \true$
iff there is $S \in f(x)(a)$ such that $\tr(y)(w) = \true$ for all $y \in S$ (see Appendix~\ref{sec:details-examples}).

\end{example}

\begin{example}[Weighted tree automata]\label{ex:wta}
In this example we let $\C = \D = \Set$ and $F = G = \sem^-$ for a semiring $\sem$.
We consider coalgebras for $\M \Sigma$ (Section~\ref{sec:prelim}), where $\Sigma$ is a polynomial functor corresponding
to a signature. The initial algebra of $\Sigma$ is carried by
the set of finite $\Sigma$-trees, denoted by $\Sigma^* \emptyset$.	
Define $\rho \colon \Sigma F \Rightarrow F \Sigma$ by cases on the operators $\sigma$ in the signature:
$$
\rho_X(\sigma(\varphi_1, \ldots, \varphi_n))(\tau(x_1, \ldots, x_m)) = 
\begin{cases}
	\prod_{i = 1..n}{\varphi_i(x_i)} & \text{if } \sigma = \tau \\
	0 & \text{otherwise}
\end{cases}
$$
where $n$ is the arity of $\sigma$. 
Define $\alpha \colon \sem^- \Rightarrow \sem^\M$ by its mate:
$
	\alpha^\flat_\Phi(\varphi)(w) = \sum_{\psi \in \sem^\Phi} \varphi(\psi)\cdot \psi(w)
$.
Notice that $\alpha$ and $\rho$ generalize the logics of Example~\ref{ex:nda}.

The logical semantics of $\alpha \circledcirc \rho$ on a weighted tree automaton $f \colon X \rightarrow \M \Sigma X$ 
is the unique map	$\tr \colon X \rightarrow \sem^{\Sigma^* \emptyset}$ making the following diagram commute:
\begin{equation}\label{eq:wta-trace}
\xymatrix@C=1.3cm{
	X \ar[d]_f \ar[rrr]^{\tr} & & & \sem^{\Sigma^* \emptyset} \ar[d] \\
	\M \Sigma X \ar[r]_-{\M \Sigma \tr} 
	&\M \Sigma \sem^{\Sigma^* \emptyset} \ar[r]_-{\M \rho^\flat_{\Sigma^*\emptyset}} 
		&\M\sem^{\Sigma \Sigma^* \emptyset} \ar[r]_-{\alpha^\flat_{\Sigma \Sigma^*\emptyset}}
		&\sem^{\Sigma \Sigma^* \emptyset}
}
\end{equation}
This means that for any tree $\sigma(t_1, \ldots, t_n)$ and any $x \in X$ we have (see Appendix~\ref{sec:details-examples}):
\begin{equation*}
	\tr(x)(\sigma(t_1, \ldots, t_n)) = \sum_{x_1, \ldots, x_n \in X} f(x)(\sigma(x_1, \ldots, x_n)) \cdot \prod_{i = 1..n} \tr(x_i)(t_i)
\end{equation*}
As a special case, we obtain for any \emph{weighted automaton}
$f \colon X \rightarrow \M(A \times X + 1)$ a unique map $\tr \colon X \rightarrow \sem ^ {A^*}$ so that for any $x \in X$, $a \in A$ and $w \in A^*$:
$\tr(x)(\varepsilon) = f(x)(*)$ and $\tr(x)(aw) = \sum_{y \in X} f(x)(a,y) \cdot \tr(y)(w)$.
For $\sem$ the Boolean semiring we get 
the usual semantics of tree automata:
$\tr(x)(\sigma(t_1, \ldots, t_n)) = \true$ iff there are $x_1, \ldots, x_n$ such that $\sigma(x_1, \ldots, x_n) \in f(x)$
and for all $i \leq n \colon \tr(x_i)(t_i)=\true$.

The $\Sigma$-algebra $\hat{F}(X,f)$ (see~\eqref{eq:liftF}) is a \emph{deterministic bottom-up tree automaton}.
It corresponds to the top-down automaton $f$, 
in the sense that the semantics $\tr$ of $f$ is the transpose of the unique homomorphism 
$s \colon \Sigma^*\emptyset \rightarrow \sem^X$ arising by initiality; 
the latter is the usual semantics of bottom-up tree automata.
\end{example}

\begin{example}[Probabilistic systems]\label{ex:probsys}
Consider generative probabilistic transition systems~\cite{vGSS95} with explicit termination, modeled as coalgebras $f\colon X\to \Delta(A\times X+1)$, where $\Delta$ is the finitely supported probability distribution functor on $\Set$. One would like to interpret sequences of labels from $A$ as {\em completed} traces for such coalgebras, i.e., ones ending with a transition to the unique element of $1$, and assign probabilities to them.

Although probability distributions on a set $X$ can be seen as functions from $X$ to the interval $[0,1]$, techniques of Example~\ref{ex:wta} are not directly applicable, since $[0,1]$ is not a semiring in the expected sense: addition is not a total operation. One could replace $[0,1]$ with the semiring of nonnegative real numbers and proceed as in Example~\ref{ex:wta}, but the resulting trace semantics would obscure an important property of probabilistic traces: every trace has probability at most $1$. Actually, in the example considered here, an even stronger property holds: for every process, probabilities of all complete traces generated from it form a subprobability distribution. We wish to design a forgetful logic framework that would make this property apparent.

To this end, put $\C=\Set$ and $\D=\PCM$, the category of partial commutative monoids (see~\cite{pcm} for details). A partial commutative monoid is a set $X$ together with a unit $0\in X$ and a partial binary addition function $\ovee$ on $X$ that is commutative and associative whenever defined. A morphism of such monoids is a function that preserves units and addition, whenever defined. A paradigmatic example of a partial commutative monoid is the interval $[0,1]$ with $0$ as unit and addition defined whenever the result is at most $1$.

The obvious forgetful functor $U\colon \PCM\to\Set$ has a left adjoint $(-)_{\bot}\colon \Set\to\PCM$ that assigns to every set $X$ the partial monoid $X+\bot$ with $\bot$ as the unit, and addition defined by $\bot\ovee x=x\ovee \bot=x$ and undefined otherwise.

The category $\PCM$ is symmetric monoidal closed~\cite{pcm} with the internal hom-functor $\Phi\multimap \Psi=\hom_{\PCM}(\Phi,\Psi)$ with the constant function $0(\phi)=0_{\Psi}$ as the unit, and addition defined by
\[
	(f\ovee g)(\phi) = f(\phi)\ovee g(\phi)
\]
if the addition on the right is defined for every $\phi\in\Phi$, and $f\ovee g$ undefined otherwise. The symmetric monoidal closed structure implies that the functor $-\multimap \Psi$ is contravariant self-adjoint for any partial commutative monoid $\Psi$.

Define contravariant functors
\[
	FX = X_{\bot}\multimap [0,1] 
	\qquad \mbox{and} \qquad
	G\Phi = \hom_{\PCM}(\Phi,[0,1]).
\]
They form a contravariant adjunction: a composition of the adjunction $(-)_{\bot}\dashv U$ with the contravariant self-adjunction of $-\multimap [0,1]$. In other words, there is a bijection
\[
	X \to \hom_{\PCM}(\Phi,[0,1]) \cong \hom_{\PCM}(\Phi,X_{\bot}\multimap[0,1])
\]
natural in $X\in\Set$ and $\Phi\in\PCM$.

Define $\alpha\colon F\To F\Delta$ by its mate $\alpha^{\flat} \colon \Delta G\To G$:
\[
	\alpha^{\flat}_{\Phi}(\delta)(\phi) = \underset{f\in G\Phi}{\scalebox{2}{$\ovee$}}\delta(f)\cdot f(\phi).
\]
This takes values in $[0,1]$, since
\[
	0\leq \underset{f\in G\Phi}{\scalebox{2}{$\ovee$}}\delta(f)\cdot f(\phi) \leq \underset{f\in G\Phi}{\scalebox{2}{$\ovee$}}\delta(f) = 1.
\]

$\PCM$ has products, defined as expected as cartesian products on carriers. Define a functor $L$ on $\PCM$ by
\[
	L\Phi = \Phi^A_{\omega}\times \{\top\}_{\bot}
\]
where $\Phi^A_{\omega}$, for the set $A$ of transition labels, is the $A$-fold product of $\Phi$ restricted to those tuples that have the unit $0_{\Phi}$ on all but finitely many components. If $A$ is finite then this is simply the $A$-fold product of $\Phi$.

Elements of $L\Phi$ are $A$-indexed families of elements of $\Phi$ that are $0_{\Phi}$ almost everywhere, with an additional component that is either $\top$ or $\bot$, with addition defined componentwise (note that $\top\ovee\top$ is undefined).

Consider the set $\Lambda=\pow_{\omega}(A^*)$ of all finite sets of $A$-traces, considered as a partial commutative monoid with $\emptyset$ as the unit, and set union as addition, but defined only for disjoint sets. This partial monoid carries an $L$-algebra structure $h\colon L\Lambda\to\Lambda$ defined by:
\begin{align*}
	h(T_a,T_b,\ldots,\bot) &= \{aw\mid a\in A,\ w\in T_a\} \\
	h(T_a,T_b,\ldots,\top) &=  \{aw\mid a\in A,\ w\in T_a\} \cup \{\epsilon\},
\end{align*}
for any $T_a,T_b,\ldots\subseteq A^*$. It is easy check that $h$ is a bijection and a morphism in $\PCM$. In particular, $h$ takes values in finite sets thanks to the restriction of $\Phi^A_{\omega}$ to tuples that are $\emptyset$ almost everywhere.

Moreover, $h$ is an initial $L$-algebra. To see this, consider any algebra $k\colon L\Phi\to\Phi$. The unique algebra map $f\colon\Lambda\to\Phi$ from $h$ to $k$ is defined by induction on the length of the longest trace in elements of $\Lambda$:
\begin{align*}
	f(\emptyset) &= 0_{\Phi} = k(0_{\Phi},0_{\Phi},\ldots,\bot) \\
	f(\{\epsilon\}) &= k(0_{\Phi},0_{\Phi},\ldots,\top) \\
	f(T) &= k(f(T_a),f(T_b),\ldots,\bot) && \text{if } \epsilon\not\in T \\
	f(T) &= k(f(T_a),f(T_b),\ldots,\top) && \text{if } \epsilon\in T
\end{align*}
where $T_a=\{w\in A^* \mid aw\in T\}$, for $T\subseteq A^*$ and $a\in A^*$. This is a well-formed inductive definition since each $T_a$ only contains traces strictly shorter than the longest trace in $T$. It is also a partial monoid morphism, since if $T$ and $T'$ are disjoint sets of traces then $T_a$ and $T'_a$ are disjoint for each $a\in A$, and at most one of $T$ and $T'$ contains the empty word $\epsilon$. The fact that $f$ is an algebra morphism follows directly from its definition, and its uniqueness follows by routine induction.

A logic $\rho\colon LF\To FB$, where $BX=A\times X+1$, can be defined by its mate $\rho^{\flat}\colon BG\To GL$:
\begin{align*}
	\rho^{\flat}_{\Phi} &: A\times\hom_{\PCM}(\Phi,[0,1])+1\to \hom_{\PCM}(\Phi^A\times\{\top\}_{\bot},[0,1]) \\
	\rho^{\flat}_{\Phi}(a,h)(\phi_a,\phi_b,\ldots,x) &= h(\phi_a) \qquad \mbox{for } x\in\{\top,\bot\} \\
	\rho^{\flat}_{\Phi}(\ast)(\phi_a,\phi_b,\ldots,x) &= 
		\left\{\begin{array}{l}
			0 \quad \mbox{if }x=\bot \\
			1 \quad \mbox{if }x=\top
		\end{array}\right. 
\end{align*}
where $\ast$ is the unique element of $1$. Both $\alpha$ and $\rho$ are easily seen to be natural, so this completes an instance of the framework of forgetful logics. 

As a result, for any probabilistic transition system $f\colon X\to \Delta(A\times X+1)$ we obtain a map $s^{\flat}\colon X\to \hom_{\PCM}(\Lambda,[0,1])$ that assigns, to every state $x\in X$, a map from finite sets of $A$-traces to the interval $[0,1]$ that is additive as far as disjoint sets of traces are concerned. Such a map gives (indeed, is equivalent to) a subprobability distribution $\delta$ on the set of all $A$-traces. Indeed, put $\delta(w)=s^{\flat}(\{w\})$; then
\[
	\sum_{w\in A^*}\delta(w) = \sup_{W\subseteq_{\textrm{fin}}A^*}\left(\sum_{w\in W}\delta(w)\right) = 
	\sup_{W\subseteq_{\textrm{fin}}A^*}s^{\flat}(W) \leq 1.
\]

Note that the partial monoid homomorphism $s\colon\Lambda\to (X_{\bot}\multimap[0,1])$ does {\em not} map traces to any distributions of states. 

In~\cite[Sec.~7.2]{JSS14}, where another coalgebraic approach to generative probabilistic systems was developed, the subprobability distribution monad had to be used instead of $\Delta$. The reason for that was that even though a state in a generative probabilistic system determines a full probability distribution on successors (including immediate termination, i.e.~the empty trace, as a degenerate successor), the finite complete traces of a state do not form a probability distribution, but only a subprobability distribution, due to possibility of infinite traces. In the forgetful logic approach the branching type functor $\Delta$ is distinguished from the functor $G$ used to collect a structure of traces, so the distribution monad can be used. Note that although the subdistribution monad is naturally isomorphic to the functor $\Delta(-+1)$ that is present in our behaviour functor, the $+1$ component here is not a part of the branching structure: it models trace termination and it is very much a part of the trace behaviour structure. In~\cite{JSS14}, the subdistribution monad is used on top of that, with an additional $+1$ component necessary to handle global {\em non}termination.

The technique developed in this example does not work for {\em reactive} probabilistic transition systems modeled as coalgebras $f\colon X\to [0,1]\times (\Delta X)^A$. Technically, it is not clear how to find a functor $L$ on $\PCM$ with a natural transformation $\rho\colon LF\To FB$ for $BX=[0,1]\times X^A$, that would model trace semantics as expected. This is not surprising, as under a standard probabilistic semantics of reactive systems, traces accepted from a fixed state do not form a (sub)probability distribution.
 
\end{example}

\section{Forgetful logics for monads}\label{sec:monads}

In most coalgebraic attempts to trace semantics~\cite{Cirstea13,Goncharov13,JSS14,kissigKurz,KMPS15,powerturi}, the functor $T$, which models the branching aspect of system behaviour, is assumed to be a monad. The basic definition of a forgetful logic is more relaxed in that it allows an arbitrary functor $T$ but one may notice that in all examples in Section~\ref{sec:examples}, $T$ is a monad.

In coalgebraic approaches cited above, the structure of $T$ is resolved using monad multiplication $\mu\colon TT\To T$. Forgetful logics use transformations $\alpha\colon  F\To FT$ with their mates $\alpha^\flat\colon TG\To G$ for the same purpose. If $T$ is a monad, it will be useful to assume a few basic axioms analogous to those of monad multiplication:

\begin{definition}\label{def:action}
Let $(T, \eta, \mu)$ be a monad. 
A natural transformation $\alpha^\flat \colon T G \Rightarrow G$ is a \emph{($T$)-action} (on $G$) if
the following diagram commutes:
$$
\xymatrix{
	TTG \ar@{=>}[d]_{\mu G} \ar@{=>}[r]^{T\alpha^\flat} & TG \ar@{=>}[d]_-{\alpha^\flat} & G \ar@{=>}[l]_-{\eta G} \ar@{=}[dl]\\
	TG \ar@{=>}[r]_{\alpha^\flat} & G &
}
$$
i.e., if each component of $\alpha^\flat$ is an Eilenberg-Moore algebra for $T$.
\end{definition}
Monad actions on functors are to monads as monoid actions on sets are to monoids.

Logics $\alpha$ whose mates are monad actions have a characterization in terms of their corresponding encodings $\alpha^\dagger:T\To GF$ defined as in Section~\ref{sec:forgetful}:
\begin{lemma}\label{lem:monadmap}
For any $\alpha \colon F\To FT$, the mate $\alpha^\flat \colon TG\To G$ is a monad action if and only if $\alpha^\dagger \colon T\To GF$ is a monad morphism.
\end{lemma}
\begin{proof}
This is a special case of~\cite[Prop.~II.1.4]{dubuc}, but see Appendix~\ref{app:sec5} for a self-contained proof.
\end{proof}

This means that a forgetful logic for $T$ whose mate is a monad action, is equivalent to an encoding of the monad structure of $T$ (which is used for resolving branching in~\cite{powerturi,HJS07,JSS14,SBBR13}) into the monad $GF$. We shall use this connection in Section~\ref{sec:determinization} to relate forgetful logics to the determinization constructions of~\cite{JSS14}. 

It is easy to check by hand that in all examples in Section~\ref{sec:examples}, $\alpha^{\flat}$ is an action, but it also follows from the following considerations.

In some well-structured cases, one can search for a suitable $\alpha$ by looking at $T$-algebras in $\C$. We mention it only briefly and not explain the details, as it will not be directly used in the following.

If $\C$ has products, then for any object $V\in\C$ there is a contravariant adjunction as in Section~\ref{sec:contr-adj}, where:
$\D=\Set$, $F = \C(-,V)$ and $G = V^-$,
where $V^X$ denotes the $X$-fold product of $V$ in $\C$. (This adjunction was studied in~\cite{lenisapowerwatanabe} for the purpose of combining distributive laws.) By the Yoneda Lemma, natural transformations $\alpha\colon F\To FT$ are in bijective correspondence with algebras $g \colon TV\to V$. Routine calculation shows that the mate $\alpha^\flat$ is a $T$-action if and only if the corresponding $g$ is an Eilenberg-Moore algebra for $T$.
	
Alternatively, one may assume that $\C=\D$ is a symmetric monoidal closed category and $F=G=V^-$ is the internal hom-functor based on an object $V\in \C$. 
(This adjunction was studied in~\cite{klin07} in the context of coalgebraic modal logic.) If, additionally, the functor $T$ is strong, then every algebra $g \colon TV\to V$ gives rise to $\alpha \colon F\To FT$, whose components $\alpha_X  \colon  V^X \to V^{TX}$ are given by transposing:
\[\xymatrix{
	TX\otimes V^X\ar[rr]^{\textrm{strength}} & & T(X\otimes V^X)\ar[rr]^-{T(\textrm{application})} & & TV\ar[r]^g & V
}\]
If $T$ is a strong monad and $g$ is an Eilenberg-Moore algebra for $T$ then $\alpha^\flat$ is a $T$-action.

In an enriched setting, if $\D$ is enriched over $\C$, both these constructions are instances of a more general one based on the existence of suitable powers.	

If $\C=\D=\Set$ then both constructions apply (and coincide). 
In this situation more can be said~\cite{JSS14,Jacobs:coalg}: the resulting contravariant adjunction can be factored through the category of Eilenberg-Moore algebras for $T$. 

\section{Determinization}\label{sec:determinization}

The classical powerset construction turns a non-deterministic automaton into a deterministic one, with states of the former interpreted as singleton states in the latter.
More generally, a determinization procedure of coalgebras involves a change of state space.
We define it as follows.
\begin{definition}\label{def:det}
Let $T, H \colon \C \rightarrow \C$ be endofunctors.
A $(T)$-\emph{determinization procedure} of $H$-coalgebras consists of a natural transformation $\eta \colon \Id\To T$,
a functor $K \colon \C \rightarrow \C$ and a lifting of $T$:
$$
\xymatrix{
\coalg{H} \ar[d] \ar[r]^{\bar{T}} & \coalg{K} \ar[d] \\
\C \ar[r]_T & \C
}
$$
\end{definition}
We will mostly focus on cases where $H=TB$ or $H=BT$, but in Section~\ref{sec:isomates} we will consider situations where $T$ is not directly related to $H$.

The classical powerset construction is {\em correct}, in the sense that the language semantics of a state $x$ in a non-deterministic automaton coincides with the final semantics (the accepted language) of the singleton of $x$ in the determinized automaton.
At the coalgebraic level, we capture trace semantics by a forgetful logic. Then, a determinization procedure is correct if logical equivalence 
on the original system coincides with behavioural equivalence on the determinized system along $\eta$, formally captured
as follows.
\begin{definition}\label{def:correct} 
	A $T$-determinization procedure
	$(\bar{T},K,\eta)$ of $H$-coalgebras is \emph{correct} wrt.~a 
	logic for $H$ if, for any $H$-coalgebra $(X,f)$ with logical semantics $\tr$:
	\begin{enumerate}
		\item $\tr$ factors through $h \circ \eta_X$, for any $K$-coalgebra homomorphism $h$ from $\bar{T}(X,f)$. 
		\item there exists a $K$-coalgebra homomorphism $h$ from $\bar{T}(X,f)$ and a mono $m$ so that $\tr = m \circ h \circ \eta_X$.
	\end{enumerate}\medskip
\end{definition}

\noindent The first condition states that behavioural equivalence on the determinized system implies logical equivalence
on the original system; the second condition states the converse. 

It is standard to define a lifting $\bar{T} \colon \coalg{H} \rightarrow \coalg{K}$ of $T$ from a natural transformation
$\lambda \colon TH \Rightarrow KT$, as follows:
\begin{equation}\label{eq:lifting-from-transformation}
\bar{T}(X,f) = (TX, \lambda_X \circ Tf) \qquad \bar{T}(h) = Th
\end{equation}
for any $H$-coalgebra $(X,f)$ and coalgebra morphism $h$. 
All of the (liftings in) determinization constructions considered in this paper arise from such natural transformations.

In~\cite{JSS14} a more specific kind of determinization for $TB$-coalgebras was studied, arising
from a natural transformation $\kappa \colon TB \Rightarrow KT$ and a monad $(T, \eta, \mu)$. 
This construction is an instance of~\eqref{eq:lifting-from-transformation}, by setting $H = TB$ and 
$\lambda = \kappa \circ \mu B : TTB \Rightarrow KT$. We denote the lifting of $T$ arising in this way by $T^\kappa$.
Spelling out the details,
for any $TB$-coalgebra $(X,f)$ we have:
\begin{equation}\label{eq:det-tb}
  \xymatrix{T^\kappa(X,f) = (TX \ar[r]^-{Tf} & TTBX \ar[r]^-{\mu_{BX}} & TBX \ar[r]^-{\kappa_X} & KTX)}.
\end{equation}
For examples see, e.g.,~\cite{JSS14} and the end of this section.

The same type of natural transformation can be used to determinize $BT$-coalgebras. 
This is again an instance of~\eqref{eq:lifting-from-transformation}, where $H = BT$ and
$\lambda = K\mu \circ \kappa T : TBT \Rightarrow KT$. We denote the lifting of $T$ arising in this way
by $T_\kappa$. Spelling out the details, for any $BT$-coalgebra $(X,f)$ we have:
\begin{equation}\label{eq:det-bt}
  \xymatrix{T_\kappa(X,f) = (TX \ar[r]^-{Tf} & TBTX \ar[r]^-{\kappa_{TX}} & KTTX \ar[r]^-{K\mu_{X}} & KTX)}.
\end{equation}
This is considered in~\cite{SBBR13,JSS14} for the case where $B=K$ and 
$\kappa$ is a distributive law of monad over functor. 

In Theorem~\ref{thm:det-translate} below, we give a sufficient condition for the logical semantics on $TB$ or $BT$-coalgebras 
to coincide with a logical semantics on determinized $K$-coalgebras, for the determinization constructions $T^\kappa$ and
$T_\kappa$. The same condition was recently studied in~\cite[Lemma 5.11]{corina15}, where it (more precisely, its mate) was proved to hold, under some additional assumptions, if $\alpha$ arises from the algebra $\mu_1:TT1\to T1$ as described in Section~\ref{sec:monads}. First, we prove a general result, in the setting of~\eqref{eq:lifting-from-transformation}, which relates the logical semantics of an $H$-coalgebra to a logical semantics for the $K$-coalgebra obtained
by applying the lifting $\bar{T}$. 

\begin{lemma}\label{lm:log-sem-determinize}
	Assume natural transformations $\lambda \colon TH \Rightarrow KT$,
	$\delta \colon LF \Rightarrow FH$, $\theta \colon LF \Rightarrow FK$
	and $\alpha \colon F \Rightarrow FT$ such that
	$\lambda$ is a morphism of logics from $\theta \circledcirc \alpha$ to $\alpha \circledcirc \delta$.
	
	Let $s^\flat$ be the logical semantics of $\delta$ on a coalgebra
	$f \colon X \rightarrow HX$. Then $\alpha^\flat_\Phi \circ Ts^\flat$ is
	the logical semantics of $\theta$ on the coalgebra $(TX,\lambda_X \circ Tf)$.
\end{lemma}
\begin{proof}
	Consider the following diagram:
	$$
	\xymatrix@C=1.2cm{
		TX \ar[rr]^{Ts^\flat} \ar[d]_{Tf}
			& 
			& TG\Phi \ar[d]^{TGa} \ar[r]^{\alpha^\flat_\Phi}
			& G\Phi \ar[d]^{Ga} \\
		THX \ar[d]_{\lambda_X} \ar[r]^{THs^\flat} 
			& THG\Phi \ar[d]^{\lambda_{G\Phi}} \ar[r]^{T\delta^\flat_\Phi}
			& TGL\Phi \ar[r]^{\alpha^\flat_{L\Phi}}
			& GL\Phi \ar@{=}[d] \\
 		KTX \ar[r]_{KTs^\flat}
 			& KTG\Phi \ar[r]_{K\alpha^\flat_\Phi}
 			& KG\Phi \ar[r]_{\theta^\flat_\Phi}
 			& GL\Phi
	}
	$$
	which commutes, clockwise starting from the top left: by definition of the logical semantics $s^\flat$,
	naturality of $\alpha^\flat$, the assumption that $\lambda$ is a morphism of logics,
	and naturality of $\lambda$. Commutativity of the outside of the diagram
	means that $\alpha^\flat_\Phi \circ Ts^\flat$ is indeed the logical
	semantics of $\theta$ on $(TX,\lambda_X \circ Tf)$.
\end{proof}

\begin{theorem}\label{thm:det-translate}
Assume:
\begin{itemize}
	\item a monad $(T, \eta, \mu)$;
	\item a forgetful logic $\alpha \colon F \Rightarrow FT$, $\rho \colon LF \Rightarrow FB$, such
	that $\alpha^\flat$ is a $T$-action on $G$ (Definition~\ref{def:action});
	\item a natural transformation $\kappa \colon TB \Rightarrow KT$, and
	\item another logic $\theta \colon LF\To FK$,
\end{itemize}
such that $\kappa$ is a morphism of logics from $\theta \circledcirc \alpha$
to $\alpha \circledcirc \rho$, i.e., the following diagram commutes:
$$
\xymatrix{
	TBG \ar@{=>}[r]^{T\rho^\flat} \ar@{=>}[d]_{\kappa G}
		& TGL \ar@{=>}[r]^{\alpha^\flat L} 
		& GL \ar@{=}[d]\\
	KTG \ar@{=>}[r]_{K\alpha^\flat} & KG \ar@{=>}[r]_{\theta^\flat} 
	& GL.
}
$$
Let $\tr$ be the semantics of $\alpha \circledcirc \rho$ on a coalgebra $f \colon X \rightarrow TBX$,
and let $\trr{\theta}$ be the semantics of $\theta$ on $T^\kappa(X,f)$ (see~\eqref{eq:det-tb}). 
Then $\tr = \trr{\theta} \circ \eta_X$.

The same holds for the determinization procedure $T_\kappa$ (see~\eqref{eq:det-bt}) for $BT$-coalgebras and the logic $\rho \circledcirc \alpha$.
\end{theorem}
\begin{proof}
	We first consider $TB$-coalgebras. Our aim is to use Lemma~\ref{lm:log-sem-determinize},
	instantiated to $H=TB$, $\delta = \alpha \circledcirc \rho$ and $\lambda = \kappa \circ \mu B$.
	To this end, we must show that $\lambda$ is a morphism of logics from $\theta \circledcirc
	\alpha$ to $\alpha \circledcirc \alpha \circledcirc \rho$, i.e., the outside
	of the following diagram should commute:
	$$
	\xymatrix{
		TTBG \ar@{=>}[r]^{TT\rho^\flat} \ar@{=>}[d]_-{\mu BG}
			& TTGL \ar@{=>}[d]^-{\mu{GL}} \ar@{=>}[r]^{T\alpha^\flat L}
			& TGL \ar@{=>}[d]^-{\alpha^\flat L} \\
		TBG \ar@{=>}[d]_{\kappa G} \ar@{=>}[r]^{T\rho^\flat} 
			& TGL \ar@{=>}[r]^{\alpha^\flat L}
			& GL \ar@{=}[d] \\
		KTG \ar@{=>}[r]_{K\alpha^\flat} 
			& KG \ar@{=>}[r]_{\theta^\flat} 
			& GL
	}
	$$
	Indeed, the diagram commutes, clockwise starting from the top left: by naturality of $\mu$,
	by the assumption that $\alpha^\flat$ is a $T$-action, and by the assumption that $\kappa$
	is a morphism of logics from $\theta \circledcirc \alpha$ to $\alpha \circledcirc \rho$
	(in fact, since $\alpha^\flat$ is a $T$-action, $\mu$ is a morphism of logics from $\alpha$ to $\alpha \circledcirc \alpha$,
	hence $\mu B$ is a morphism from $\alpha \circledcirc \rho$ to $\alpha \circledcirc \alpha \circledcirc \rho$, and the above diagram is obtained by composing this morphism with $\kappa$).

	By Lemma~\ref{lm:log-sem-determinize}, for any $TB$-coalgebra $(X,f)$
	we obtain 
	$$
		s^\flat_{\theta} = \alpha_\Phi^\flat \circ Ts^\flat
	$$ 
	where $s^\flat$ is the logical semantics of $\alpha \circledcirc \rho$
	on $(X,f)$ and $s^\flat_{\theta}$ 
	is the logical semantics of $\theta$ on $T^\kappa(X,f) = (TX, \kappa_X \circ \mu_{BX} \circ Tf)$.
	Hence
	$	s^\flat_{\theta} \circ \eta_X 
			= \alpha^\flat_{\Phi} \circ Ts^\flat \circ \eta_X 
			= \alpha^\flat_\Phi \circ \eta_{G\Phi} \circ s^\flat 
			= s^\flat 
	$
	where the second equality holds by naturality of $\eta$,
	and the third since $\alpha^\flat$ is a $T$-action.
	
	For $BT$-coalgebras, we instantiate Lemma~\ref{lm:log-sem-determinize} to
	$H=BT$, $\delta = \rho \circledcirc \alpha$ and $\lambda = K\mu \circ \kappa T$. We prove
	that $\lambda$ is a morphism of logics:
	$$
	\xymatrix{
		TBTG \ar@{=>}[r]^{TB\alpha^\flat} \ar@{=>}[d]_{\kappa TG}
			& TBG\ar@{=>}[d]^{\kappa G} \ar@{=>}[r]^{T\rho^\flat}
			& TGL \ar@{=>}[dd]^{\alpha^\flat L} \\
		KTTG \ar@{=>}[d]_{K \mu G} \ar@{=>}[r]^{K T \alpha^\flat}
			& KTG \ar@{=>}[d]^{K\alpha^\flat}
			& \\
		KTG \ar@{=>}[r]_{K \alpha^\flat}
			& KG \ar@{=>}[r]_{\theta^\flat}
			& GL
	}
	$$
	The diagram commutes, clockwise starting from the top left: by naturality of $\kappa$, the assumption that
	by $\kappa$ is a morphism of logics and by the assumption that $\alpha^\flat$ is a $T$-action
	(again, the above diagram amounts to a composition of the logic morphisms $\mu$ and $\kappa$).
	By Lemma~\ref{lm:log-sem-determinize}, for any $BT$-coalgebra $(X,f)$
	we obtain 
	$
		s^\flat_{\theta} = \alpha_\Phi^\flat \circ Ts^\flat
	$ 
	where $s^\flat$ is the logical semantics of $\rho \circledcirc \alpha$
	on $(X,f)$ and $s^\flat_{\theta}$ 
	is the logical semantics of $\theta$ on $T_\kappa(X,f) = (TX, K\mu_X \circ \kappa_{TX} \circ Tf)$.
	The conclusion of the proof is analogous to the above case of $TB$-coalgebras.
\end{proof}

We apply Theorem~\ref{thm:det-translate} to obtain a sufficient condition for a determinization construction
to be correct with respect to a trace semantics given by a forgetful logic. The main 
idea is to choose $\theta$ to be an expressive logic for $K$, so that logical equivalence
coincides with behavioural equivalence.
\begin{corollary}\label{cor:correctness}
	Let $(T, \eta, \mu)$, $\alpha$, $\rho$, $\theta$ and $\kappa$ be as in Theorem~\ref{thm:det-translate}, and 
	suppose that $\theta$ is an expressive logic. Then 
	the determinization procedure $T^\kappa$ of $TB$-coalgebras \eqref{eq:det-tb} is correct with respect to $\alpha \circledcirc \rho$,
	and the determinization procedure $T_\kappa$ of $BT$-coalgebras \eqref{eq:det-bt} is correct with respect to $\rho \circledcirc \alpha$.
\end{corollary}
\begin{proof}
	Let $\bar{T}$ be either $T^\kappa$ or $T_\kappa$, 
	let $(X,f)$ be a $TB$-coalgebra or a $BT$-coalgebra respectively,
	and $\tr$ the semantics of the forgetful logic on $f$. Under the above assumptions, by Theorem~\ref{thm:det-translate}
	we have $\tr = \trr{\theta} \circ \eta_X$, where $\trr{\theta}$ is the logical semantics 
	of $\theta$ on $\bar{T}(X,f)$.
	Since $\trr{\theta}$ is a logical semantics it factors through any coalgebra homomorphism,
	yielding condition (1) of correctness,
	and since it is expressive it decomposes as a coalgebra homomorphism followed by a mono,
	yielding condition (2).
\end{proof}

To illustrate all this, we show that the determinization 
of weighted automata as given in~\cite{JSS14}
is correct with respect to weighted language equivalence. (There is no such result for tree automata, as they do not determinize.)
\begin{example}\label{ex:wa-correct}
	Fix a semiring $\sem$, let $B = A \times - +1$ and $K = \sem \times -^A$.
	Consider $\kappa \colon \M B \Rightarrow K \M$ defined as follows~\cite{JSS14}:
	$
		\kappa_X(\varphi) = (\varphi(*), \lambda a. \lambda x. \varphi(a,x))
	$.
	This induces a determinization procedure $\M^\kappa$ as in~\eqref{eq:det-tb}, for weighted automata.
	Let $\alpha \circledcirc \rho$ be the forgetful logic for weighted automata introduced in 
	Example~\ref{ex:wta}, and recall that the logical semantics on a weighted automaton
	is the usual notion of acceptance of weighted languages.
	We use Corollary~\ref{cor:correctness} to prove that the determinization
	procedure $\M^\kappa$ is correct with respect to $\alpha \circledcirc \rho$.
	To this end, consider the logic 
	$\theta^\flat \colon \sem \times (\sem^-)^A \Rightarrow \sem^{A \times - + 1}$ given by the isomorphism,
	similar to the logic in Example~\ref{ex:nda-bt}.
	Since $\theta^\flat$ is componentwise injective, $\theta$ is expressive. 
	Moreover, $\alpha^\flat$ is an action (see Section~\ref{sec:monads}).
	The only remaining condition is commutativity of the diagram in Theorem~\ref{thm:det-translate}, which is a straightforward calculation.
	This proves correctness of the determinization $\M^\kappa$ with respect to the semantics of $\alpha \circledcirc \rho$.
\end{example}
\begin{example}\label{ex:nda-correct}
	In~\cite{SBBR13} it is shown how to
	determinize non-deterministic automata of the form $B\powf$, 
	where $BX = 2 \times X^A$, based on $\kappa = \langle \kappa^o, \kappa^t \rangle \colon \powf(2 \times -^A) \Rightarrow 2 \times (\powf-)^A$ 
	(note that $B=K$ in this example)
	where 
		$\kappa^o_X(S) = \true$ iff $\exists t. (\true, t) \in S$, and $\kappa^t_X(a) = \{x \mid x \in t(a)$ for some $(o,t) \in S\}$.
  In Example~\ref{ex:nda-bt} we have seen an expressive logic $\rho$ and an $\alpha$ so that the logical semantics of $\rho \circledcirc \alpha$ yields 
	the usual language semantics. It is now straightforward to check that the determinization $\kappa$ together with the logics $\rho$, $\alpha$ 
	above satisfies the condition of Theorem~\ref{thm:det-translate}, where $\theta = \rho$. 
	By Corollary~\ref{cor:correctness} this shows the expected result that determinization of non-deterministic automata is correct with respect to language semantics.
	
  Moreover, recall that the logic $\rho \circledcirc \beta$, where $\beta$ is as defined in Example~\ref{ex:nda-bt}, yields a conjunctive semantics. 
	Take the natural transformation
	$\tau = \langle \tau^o, \tau^t \rangle$ of the same type as $\kappa$, where $\tau^o(S) = \true$ iff $o=\true$ for every $(o,t) \in S$, and $\tau^t = \kappa^t$.
	Using Corollary~\ref{cor:correctness} we can verify that this determinization procedure is correct.	
\end{example}
One can also get the finite trace semantics of transition systems (Example~\ref{ex:nda-bt}) by turning them into non-deterministic automata (then, $B$ and $K$
are different).


\begin{example}\label{ex:alt-correct}
	Alternating automata (see Example~\ref{ex:alt-aut}) can be determinized into non-deterministic automata in a process that preserves the language semantics. We shall now see how this arises as an application of Theorem~\ref{thm:det-translate} without a reference to final semantics of coalgebras.
	
	This problem is more subtle than it may look, and our solution from a previous version of this paper~\cite{KlinR15} suffered from a serious mistake. We therefore present a corrected solution a little more elaborately.
	
	Let $BX = 2 \times X^A$ as before. We wish to determinize $B\powf\powf$-coalgebras where, according to Example~\ref{ex:alt-aut}, the outer $\powf$ is interpreted disjunctively (as in nondeterministic automata), and the inner $\powf$ is interpreted conjunctively (as in co-nondeterministic automata). The result should be a nondeterministic automaton, i.e., a $B\powf$-coalgebra with the $\powf$ interpreted disjunctively. If the original alternating automaton had a carrier $X$, then the nondeterministic automaton should have a carrier $\powf X$, interpreted conjunctively. 
	
To model all this, we should instantiate Theorem~\ref{thm:det-translate} so that $T=\powf$, the functors $B$ and $K$ from the theorem are both $B\powf = 2 \times (\powf-)^A$, the logic $\alpha$ from the theorem is instantiated to $\beta$ from Example~\ref{ex:alt-aut} and the logics $\rho$ and $\theta$ from the theorem are both instantiated to $\rho \circledcirc \alpha$ from Example~\ref{ex:alt-aut}. 

We then need to provide an appropriate natural transformation $\kappa \colon \powf B\powf\To B\powf\powf$ that would model the intended determinization procedure. It is natural to define it as a composition of two steps:
\[
  \kappa = B\chi \circ \tau\powf : \powf B \powf \Rightarrow B \powf \powf
\]
where $\tau \colon \powf B\To B \powf$ is defined as at the end of Example~\ref{ex:nda-correct}. As argued there, this distributes $B$ over $\powf$ according to the conjunctive interpretation modeled by the logic $\beta$.
	
The question remains how to define $\chi \colon  \powf \powf \Rightarrow \powf \powf$ so that it distributes the ``conjunctive'' powerset over the ``disjunctive'' one as intended. One may attempt (as we did in~\cite{KlinR15}):
\begin{equation}\label{eq:wrongChi}
	\chi_X(S) = \{ \overrightarrow{g}(S) \mid g\colon S\to X \mbox{ s.t. } g(U)\in U \mbox{ for each } U\in S\}
\end{equation}
	which, given a family of sets $S$, returns all possible sets obtained by choosing a single element from each set in $S$.
	
	Unfortunately this is {\em not} a natural transformation. The following counterexample is due to Joost Winter: consider
\begin{gather*} 
X=\{a,b,c\} \qquad Y=\{d,e\} \\
f \colon X\to Y \qquad f(a)=f(b)=d \qquad f(c)=e
\end{gather*}
and calculate:
\begin{align*}
(\powf\powf f)(\chi_X(\{\{a,c\},\{b,c\}\})) &= (\powf\powf f)(\{\{a,b\},\{a,c\},\{b,c\},\{c\}\}) = \{\{d\},\{d,e\},\{e\}\} \\
\chi_Y((\powf\powf f)(\{\{a,c\},\{b,c\}\})) &= \chi_Y(\{\{d,e\}\}) = \{\{d\},\{e\}\}
\end{align*}
therefore the naturality square for $\chi$ on $f \colon X\to Y$ does not commute for the argument 
$\{\{a,c\},\{b,c\}\}\in\powf\powf X$.

This mistake has occured in the literature before. In~\cite[pp. 183-184]{manes-mulry},~\eqref{eq:wrongChi} was claimed to be a distributive law of the powerset monad over itself, and a monad structure on the double (covariant) powerset was derived from that claim in the standard way. The above counterexample applies there, and the structure defined in~\cite{manes-mulry} is not, in fact, a monad. The mistake can be traced back to~\cite[pp. 76-79]{manes}, where the same purported monad structure is defined via a Kleisli triple which, in fact, fails to satisfy one of the axioms of Kleisli triples.

To avoid these problems, we choose another definition of $\chi$: given a finite family $S$ of finite sets, it returns all sets that:
\begin{itemize}
\item are contained in the union of $S$ (and are therefore finite, and there are finitely many of them), and
\item intersect with every set in $S$.
\end{itemize}
Formally:
\begin{equation}\label{eq:goodChi}
	\chi_X(S) = \left\{V\subseteq \bigcup S \mid V\cap U\neq\emptyset \mbox{ for each }U\in S\right\}.
\end{equation}
This turns out to be a natural transformation (see Appendix~\ref{app:sec6}). As a side remark, note that $\chi$ is {\em not} a distributive law of the monad $\powf$ over itself; it does not even satisfy the unit axioms.

	Now the composition $\kappa=B\chi \circ \tau\powf$ yields a determinization procedure. To show that the diagram in Theorem~\ref{thm:det-translate} commutes, a key ingredient is the fact that $\chi$ correctly distributes conjunction over disjunction or, more formally, that the diagram
\begin{equation}\label{eq:chialphabeta}
\vcenter{\xymatrix{
	\powf\powf G\ar@{=>}[r]^-{\powf\alpha^\flat}\ar@{=>}[d]_-{\chi G} & \powf G\ar@{=>}[rd]^-{\beta^\flat} \\
	\powf\powf G\ar@{=>}[r]_{\powf\beta^\flat} & \powf G\ar@{=>}[r]_{\alpha^\flat} & G
}}
\end{equation}
commutes, where $\alpha$ and $\beta$ are as in Example~\ref{ex:alt-aut}. (See Appendix~\ref{app:sec6} for details.)

	By Theorem~\ref{thm:det-translate} we conclude that for any alternating automaton: $\tr = \trr{\rho \circledcirc \alpha} \circ \eta_X$
	where $X$ is the set of states, $\tr$ is the language semantics of the alternating automaton as in Example~\ref{ex:alt-aut}, and $\trr{\rho \circledcirc \alpha}$ is the usual language semantics (as in Example~\ref{ex:nda-bt}) of the 
	non-deterministic automaton obtained by determinization.	
	
The determinization procedure formalized by our $\chi$ (and therefore $\kappa$) is slightly different from the standard procedure for transforming alternating automata into nondeterministic ones. It still returns an automaton where states are sets of states of the original alternating automaton, but the family of successors of each such state is larger than in the standard procedure (in particular, it is closed under taking set-theoretic unions). This does not change the language semantics (additional successor states, being supersets of some successor states in the standard nondeterministic automaton, do not accept any additional words), but the reachable part of the nondeterministic automaton may be larger than the one obtained by the standard procedure. As a result, although our determinization is correct, it may be less efficient than the standard one. This is the price we paid to salvage the naturality of $\chi$. We do not know how to model precisely the standard transformation of alternating automata into nondeterministic ones in our framework.
	
\end{example}

\section{Logics whose mates are isomorphisms}\label{sec:isomates}

Corollary~\ref{cor:correctness} provides a sufficient condition for a given determinization procedure to be correct with respect to a forgetful logic. However, in general there is no guarantee that a correct determinization procedure for a given logic exists. Indeed it would be quite surprising if it did: the language semantics of (weighted) tree automata (see Example~\ref{ex:wta}) is an example of a forgetful logic, and such automata are well known not to determinize in a classical setting.

In this section we provide a sufficient condition for a correct determinization procedure to exist. 
Specifically, for an endofunctor $B$, 
we assume a logic $\rho$ whose mate 
$\rho^{\flat} \colon  BG \To GL$ is a natural isomorphism. 
This condition holds, for instance, for $\rho$ in Example~\ref{ex:nda-bt} and for $\theta$ in Example~\ref{ex:wa-correct}.
%
%
It has been studied before in the context of determinization constructions~\cite{JSS14}. Its important 
consequence is that the logical semantics $\tr$ in~\eqref{eq:logsem2} from Section~\ref{sec:forgetful} can be seen as a $B$-coalgebra morphism: 
\begin{equation}\label{eq:logsem3}
\vcenter{\xymatrix{
	BX\ar[r]^{Bs^{\flat}} & BG\Phi \\
	& GL\Phi\ar[u]_{(\rho^{\flat}_{\Phi})^{-1}} \\
	X\ar[uu]^h\ar[r]_{s^{\flat}} & G\Phi\ar[u]_{Ga}
}}
\end{equation}
Moreover, as shown in~\cite[Lemma 6]{JSS14} (see also~\cite{Hermida97structuralinduction}),
the construction mapping 
\[
	\xymatrix{LA \ar[r]^-g & A} \qquad \text{ to } \qquad \xymatrix{GA \ar[r]^-{Gg} & GLA \ar[r]^-{(\rho^{\flat}_A)^{-1}} & BGA}
\]
defines a functor $\hat{G}\colon \alg{L}\to \coalg{B}$, which is a contravariant adjoint to $\hat{F}$ (see~\eqref{eq:liftF} in Section~\ref{sec:forgetful}). As a result, $\hat{G}$ maps initial objects to final ones, hence $\hat{G}(\Phi,a) = (G\Phi, (\rho_{\Phi}^\flat)^{-1}\circ Ga)$ is a final $B$-coalgebra. Therefore,
$\tr$ is a final coalgebra morphism from $(X,h)$.





\subsection{Canonical determinization}\label{sec:can-det}

The setting of a forgetful logic $\alpha, \rho$ where the mate of $\rho$ is a 
natural isomorphism gives rise to the following diagrams: 
\begin{equation}
\begin{gathered}\label{eq:lifting-situation}
\xymatrix{
	\coalg{TB} \ar[r]^-{\tilde{F}} \ar[d] 
	& \alg{L} \ar@/^5px/[r]^{\hat{G}} \ar[d] 
	& \coalg{B} \ar[d] \ar@/^5px/[l]^{\hat{F}} \\
	\C \ar[r]_F
	& \D \ar@/^5px/[r]^G
	& \C \ar@/^5px/[l]^F
}
\qquad 
\xymatrix{
	\coalg{BT} \ar[r]^-{\tilde{F}'} \ar[d] 
	& \alg{L} \ar@/^5px/[r]^{\hat{G}} \ar[d] 
	& \coalg{B} \ar[d] \ar@/^5px/[l]^{\hat{F}} \\
	\C \ar[r]_F
	& \D \ar@/^5px/[r]^G
	& \C \ar@/^5px/[l]^F
}
\end{gathered}
\end{equation}
The functor $\tilde{F}$ arises from the logic $\alpha \circledcirc \rho$,
the functor $\tilde{F}'$ arises from $\rho \circledcirc \alpha$,
the functor $\hat{F}$ arises from $\rho$ and its contravariant adjoint $\hat{G}$ from
the fact that $\rho^\flat$ is iso. Note that we make no assumptions on $\alpha$; in particular, $\alpha^\flat$ need not be an action.

The composition $\hat{G} \tilde{F}$ is a determinization procedure, turning a coalgebra $f \colon X \rightarrow TBX$ into
a $B$-coalgebra with carrier $GFX$. Explicitly, $\hat{G} \tilde{F}(X,f)$ is
\begin{equation}\label{eq:hatGtildeF}
\xymatrix@C=1.2cm{
	GFX \ar[r]^-{GFf} 
	& GFTBX \ar[r]^-{G \alpha_{BX}}
	& GFBX \ar[r]^-{G \rho_X}
	& GLFX \ar[r]^-{(\rho^\flat)^{-1}_{FX}}
	& BGFX \,.
}
\end{equation}
Similarly, the composition $\hat{G} \tilde{F}'$ is a determinization procedure, mapping a coalgebra $g \colon X \rightarrow BTX$ to
\begin{equation}\label{eq:hatGtildeFprime}
\xymatrix@C=1.2cm{
	GFX \ar[r]^-{GFg}
	& GFBTX \ar[r]^-{G\rho_{TX}}
	& GLFTX \ar[r]^-{GL\alpha_{X}}
	& GLFX \ar[r]^-{(\rho^\flat)^{-1}_{FX}}
	& BGFX \,.
}
\end{equation}
These determinization procedures are correct with respect to the logics
$\alpha\circledcirc\rho$ and $\rho \circledcirc \alpha$ respectively, in the following sense, much stronger than required by Definition~\ref{def:correct}.
\begin{theorem}\label{thm:detcorr}
 For any $TB$-coalgebra $(X,f)$, the logical semantics $\tr$ of $\alpha \circledcirc \rho$ on $(X,f)$ coincides with the final semantics of the $B$-coalgebra $\hat{G} \tilde{F}(X,f)$ precomposed with the unit $\iota \colon \Id \Rightarrow GF$. The same holds for $BT$-coalgebras and $\rho \circledcirc \alpha$.
\end{theorem}
\begin{proof}
Let $(X,f)$ be a $TB$-coalgebra (the case of $BT$-coalgebras is analogous).
The logical semantics $\tr$ is defined by $\tr = Gs \circ \iota_X$, where $s \colon (L,a) \rightarrow \tilde{F}(X,f)$ is the unique
algebra morphism arising by initiality. Hence, we only need to show that
$Gs$ is the final semantics of $\hat{G} \tilde{F}(X,f)$. To see this, apply $\hat{G}$ to $s$ to get a coalgebra morphism
$$
Gs \colon \hat{G} \tilde{F}(X,f) \rightarrow \hat{G}(L,a) \, .
$$
But $\hat{G}(L,a)$ is a final coalgebra, so we are done.
\end{proof}

\subsection{Determinization after preprocessing}\label{sec:det-preproc}

Strictly speaking, the above constructions $\hat{G} \tilde{F}$ and $\hat{G} \tilde{F}'$ are not examples of 
determinization procedures as understood in~\cite{JSS14}: the functors $\hat{G}\tilde{F}$ and $\hat{G}\tilde{F'}$ lift $GF$ rather than $T$, and the liftings do 
not arise from a natural transformation as described in Section~\ref{sec:determinization}. However, they are {\em almost} examples: 
after an encoding of $TB$-coalgebras as $GFB$-coalgebras and of $BT$-coalgebras as $BGF$-coalgebras, 
they arise from a distributive law $\kappa \colon GFB\To BGF$.

Indeed, recall from Section~\ref{sec:forgetful} that $\alpha \colon F\To FT$ uniquely determines a natural transformation $\alpha^\dagger \colon T\To GF$.
Furthermore, this induces functors $\coalg{BT} \xrightarrow{\coalgfun{B\alpha^\dagger}} \coalg{BGF}$
and  $\coalg{TB} \xrightarrow{\coalgfun{\alpha^\dagger B}} \coalg{GFB}$, by simple composition.

Now, define a canonical logic for $GF$ by 
\[
	\varepsilon = \epsilon F : F \To FGF, \qquad \mbox{or equivalently,} \qquad \varepsilon^{\flat} = G\epsilon:GFG \To G.
\] 
(Note that $\varepsilon^\flat$ is {\em always} a $GF$-action on $G$.)
Given a logic $\rho$ for $B$ this gives rise to forgetful logics $\rho \circledcirc \varepsilon$ and $\varepsilon \circledcirc \rho$,
for $GFB$-coalgebras and $BGF$-coalgebras respectively. Note that these logics do not depend on $\alpha$ anymore.
\begin{lemma}\label{lem:preprocessing}
For any $\rho \colon LF \Rightarrow FB$ and $\alpha \colon F \Rightarrow FT$, 
the natural transformation $\alpha^\dagger B$ is a morphism of logics (with $\id$ as the translation of syntax) 
from $\varepsilon \circledcirc \rho$ to $\alpha \circledcirc \rho$, and $B \alpha^\dagger$ is a morphism 
from $\rho \circledcirc \varepsilon$ to $\rho \circledcirc \alpha$.
\end{lemma}
\begin{proof} 
It suffices to prove that $\alpha^\dagger$ is a morphism of logics from 
$\varepsilon$ to $\alpha$. This, recalling the definition of $\alpha^\dagger$ from Section~\ref{sec:forgetful}, means that the outside of the following diagram should commute:
$$
\xymatrix@C=1.5cm{
	TG \ar@{=}[r]  \ar@{=>}[d]_{T\iota G}
		& TG \ar@{=>}[d]^{\alpha^\flat} \\
	TGFG \ar@{=>}[d]_{\alpha^\flat FG}  \ar@{=>}[ur]_{TG\epsilon}
		& G \\
	GFG \ar@{=>}[ur]_{G \epsilon} 
		& 
}
$$
which follows from naturality of $\alpha^\flat$ and from~\eqref{eq:counit-unit}.
\end{proof}
As a consequence of Lemma~\ref{lm:logic-morph} and the above, we have commuting diagrams
\begin{equation}\label{eq:gamma-triangles}
\begin{gathered}
\xymatrix@C=0.2cm{
	\coalg{TB} \ar[dr]_{\bar{F}} \ar[rr]^{\coalgfun{\alpha^\dagger B}} 
		& & \coalg{GFB} \ar[dl]^{\bar{F}_{\varepsilon}} \\
	& \alg{L} &
}
\qquad 
\xymatrix@C=0.2cm{
	\coalg{BT} \ar[dr]_{\bar{F}'} \ar[rr]^{\coalgfun{B \alpha^\dagger}} 
		& & \coalg{BGF} \ar[dl]^{\bar{F}'_{\varepsilon}} \\
	& \alg{L} &
}
\end{gathered}
\end{equation}
where $\bar{F}_{\varepsilon}$ is the functor defined by $\varepsilon\circledcirc\rho$ as in~\eqref{eq:liftF}, and
$\bar{F}'_{\varepsilon}$ the functor defined by $\rho \circledcirc \varepsilon$.
Hence, encoding $TB$-coalgebras as $GFB$-coalgebras does not change their logical semantics.
More precisely, for any $f \colon X\to TBX$, the map from $X$ to $G\Phi$ defined as the logical semantics of $\varepsilon\circledcirc\rho$ on 
$\alpha^\dagger_{BX}\circ f$, coincides with semantics of $\alpha\circledcirc\rho$ on $f$. A similar result holds for $BT$-coalgebras.

Thanks to the mate $\rho^{\flat}\colon BG\To GL$ being an isomorphism, the functor $GF$ has a distributive law over $B$, denoted $\kappa\colon GFB\To BGF$ and defined by:
\begin{equation}\label{eq:kappa}
\xymatrix@C=1.3cm{
	GFB\ar@{=>}[r]^{G\rho} & GLF\ar@{=>}[r]^-{(\rho^{\flat})^{-1}F} & BGF.
}
\end{equation}
\begin{remark}
The same construction is used in~\cite{Johnstone:Adj-lif} to prove that a functor lifting to the category
of Eilenberg-Moore algebras for a monad induces a distributive law. There, the right adjoint $G$ is the (covariant) forgetful
functor; having a lifting of $B$ then means having a functor $L$ with the equality $BG = GL$. 
\end{remark}
Using $\kappa$ we can apply the determinization construction from~\cite{JSS14} as described in Section~\ref{sec:determinization}, putting $K=B$,
and taking the monad on $GF$ arising from the adjunction.
\begin{lemma}\label{lem:2ndstep}
The determinization procedure $(GF)^{\kappa}$ defined as in~\eqref{eq:det-tb} is correct with respect to $\varepsilon\circledcirc\rho$.
The determinization procedure $(GF)_{\kappa}$ is correct with respect to $\rho \circledcirc \varepsilon$.
\end{lemma}
\begin{proof}
We use Corollary~\ref{cor:correctness}, where we put $T=GF$, $\alpha=\varepsilon$ defined above, $K=B$, and $\theta=\rho$. Obviously then $\theta$ is expressive, and it is easy to check that $\varepsilon^\flat=G\epsilon$ is an action. The only remaining condition is that $\kappa$ is a morphism
of logics as in Theorem~\ref{thm:det-translate}, which is the outer shape of:
\[\xymatrix{
	GFBG\ar@{=>}[r]^-{GF\rho^\flat}\ar@{=>}[d]_{G\rho G} & GFGL\ar@{=>}[d]^{G\epsilon L} \\
	GLFG\ar@{=>}[d]_{(\rho^\flat)^{-1}FG}\ar@{=>}[r]^-{GL\epsilon} & GL\ar@{=>}[d]_{(\rho^{\flat})^{-1}}\ar@{=}[rd]^{\id} \\
	BGFG\ar@{=>}[r]_{BG\epsilon} & BG\ar@{=>}[r]_{\rho^\flat} & GL.
}\]
Here, the top square commutes by Lemma~\ref{lem:useful} and the bottom square by naturality of $(\rho^\flat)^{-1}$.
\end{proof}
Altogether, the following two-step determinization procedures arise:
$$
\xymatrix@C=0.65cm{
\coalg{TB} \ar[d] \ar[r]^-{\coalgfun{\alpha^\dagger B}} & \coalg{GFB} \ar[d] \ar[r]^-{(GF)^\kappa}  & \coalg{B} \ar[d]\\
\C \ar[r]^{\Id} & \C \ar[r]^{GF} & \C 
}
\xymatrix@C=0.65cm{
\coalg{BT} \ar[d] \ar[r]^-{\coalgfun{B \alpha^\dagger}} & \coalg{BGF} \ar[d] \ar[r]^-{(GF)_\kappa}  & \coalg{B} \ar[d]\\
\C \ar[r]^{\Id} & \C \ar[r]^{GF} & \C 
}
$$
and they are correct with respect to $\alpha\circledcirc\rho$ and $\rho \circledcirc \alpha$ respectively. 

Correctness can also be proved without Corollary~\ref{cor:correctness}, using Theorem~\ref{thm:detcorr}, since the procedures 
coincide with the constructions from~\eqref{eq:hatGtildeF} and~\eqref{eq:hatGtildeFprime}:
\begin{theorem}\label{lem:2dets}
The following diagrams commute:
$$
\xymatrix{
	\coalg{TB} \ar[r]^-{\coalgfun{\alpha^\dagger B}} \ar[d]_-{\tilde{F}} & \coalg{GFB} \ar[d]^-{(GF)^\kappa} \\
	\alg{L} \ar[r]_-{\hat{G}} & \coalg{B}
}
\qquad
\xymatrix{
	\coalg{BT} \ar[r]^-{\coalgfun{B\alpha^\dagger}} \ar[d]_-{\tilde{F}'} & \coalg{BGF} \ar[d]^-{(GF)_\kappa} \\
	\alg{L} \ar[r]_-{\hat{G}} & \coalg{B}.
}
$$
with $\tilde{F}$ and $\tilde{F}'$ and $\hat{G}$ defined as in~\eqref{eq:lifting-situation}.
\end{theorem}
\begin{proof}
By Lemma~\ref{lem:preprocessing}, $\alpha^\dagger B$ and $B\alpha^\dagger$ are morphisms of logics, which means that the diagrams
in~\eqref{eq:gamma-triangles} commute. Hence, we only need to check commutativity of:
$$
\xymatrix{
	 & \coalg{GFB} \ar[d]^-{(GF)^\kappa} \ar[dl]_{\bar{F}_{\varepsilon}}\\
	\alg{L} \ar[r]_-{\hat{G}} & \coalg{B}
}
\qquad
\xymatrix{
	& \coalg{BGF} \ar[d]^-{(GF)_\kappa}  \ar[dl]_{\bar{F}'_{\varepsilon}}\\
	\alg{L} \ar[r]_-{\hat{G}} & \coalg{B}.
}
$$
For the left triangle, notice that $\bar{F}_{\varepsilon}$ maps any $f \colon X \rightarrow GFBX$ to 
$$
\xymatrix@C=1.2cm{
 LFX \ar[r]^-{\rho_X} 
 & FBX \ar[r]^-{\epsilon_{FBX}} 
 & FGFBX \ar[r]^-{Ff}
 & FX
}
$$
and applying $\hat{G}$ yields
$$
\xymatrix@C=1.2cm{
 GFX \ar[r]^-{GFf}
 & GFGFBX \ar[r]^-{G\epsilon_{FBX}} 
 & GFBX \ar[r]^-{G\rho_{X}} 
 & GLFX \ar[r]^-{(\rho^\flat)^{-1}_{FX}}
 & BGFX
}
$$
which coincides with $(GF)^\kappa(X,f)$. Note that $G\epsilon F$ is the multiplication of the monad $GF$.

For the right triangle, we compute $\tilde{F}'_{\varepsilon}(g \colon X \rightarrow BGFX)$:
$$
\xymatrix@C=1.2cm{
 LFX \ar[r]^-{L\epsilon_{FX}} 
 & LFGFX \ar[r]^-{\rho_{GFX}} 
 & FBGFX \ar[r]^-{Fg}
 & FX
}
$$
and apply $\hat{G}$, yielding:
\begin{equation}\label{eq:hatg-f}
\xymatrix@C=1.2cm{
 GFX \ar[r]^-{GFg}
 & GFBGFX \ar[r]^-{G\rho_{GFX}} 
 & GLFGFX \ar[r]^-{GL\epsilon_{FX}}
 & GLFX \ar[r]^-{(\rho^\flat)^{-1}_{FX}}
 & BGFX .
}
\end{equation}
Further, $(GF)_\kappa(g \colon X \rightarrow BGFX)$ is:
$$
\xymatrix@C=1.2cm{
	GFX \ar[r]^-{GFg}
	& GFBGFX \ar[r]^-{G\rho_{GFX}}
	& GLFGFX \ar[r]^-{(\rho^\flat)^{-1}_{FGFX}}
	& BGFGFX \ar[r]^-{BG\epsilon_{FX}}
	& BGFX
}
$$
which coincides with~\eqref{eq:hatg-f} by naturality of $(\rho^\flat)^{-1}$.
\end{proof}

\subsection{A connection to Brzozowski's algorithm}\label{sec:brzozowski}

Call a $B$-coalgebra \emph{observable} if the morphism into a final coalgebra (assuming it exists) is mono~\cite{BBHPRS14}. 
For $\D=\Set$, the above canonical determinization procedure can be adapted to construct, for any $TB$-coalgebra,
an observable $B$-coalgebra whose final semantics coincides with the logical semantics on the original one.

Indeed, recall that $\Set$ has an (epi,mono)-factorization system and $L$ (as every endofunctor on $\Set$) preserves epimorphisms. From this it follows that every $L$-algebra homomorphism decomposes as a surjective homomorphism followed by an injective one. Given a coalgebra $f \colon X \rightarrow TBX$, apply this to
decompose the algebra homomorphism $s \colon (\Phi,a) \rightarrow \tilde{F}(X,f)$ as $s = m \circ e$, 
where $m$ and $e$ are injective and surjective respectively; call the $L$-algebra in the middle $(R, r)$. 
Recall that $Gs$ is a coalgebra homomorphism into the final coalgebra. In the present situation it decomposes as follows:
$$
\xymatrix{
	\hat{G}\tilde{F}(X,f) \ar[r]_{Gm} \ar@(ur,ul)[rr]^{Gs}
	& \hat{G}(R,r) \ar[r]_{Ge} 
	& \hat{G}(\Phi, a)
}
$$
and recall that $\hat{G}(\Phi, a)$ is a final coalgebra. Because $G$ is a right adjoint, it maps epis to monos,
therefore $Ge$ is injective and $\hat{G}(R,r)$ is observable. Moreover, thanks to Theorem~\ref{thm:detcorr} we have $\tr = Gs \circ \iota_X = Ge \circ Gm \circ \iota_X$, hence the final semantics $Ge$ of $\hat{G}(R,r)$ coincides with the logical semantics on $(X,f)$ along the mapping $Gm\circ\iota_X$.

%
Note that the construction of $\hat{G}(R,r)$ from $(X,f)$ is not a determinization procedure itself according to Definition~\ref{def:det}, as it does not lift any functor on $\C$.

The above refers to $TB$-coalgebras, but as everything else in this section, analogous reasoning works also for $BT$-coalgebras.
For $T = \Id$ and $B = 2 \times -^A$ on $\C=\Set$, that (almost) corresponds to Brzozowski's algorithm for
minimization of deterministic automata~\cite{brzozowski}. Applying $\tilde{F}$ to the given automaton corresponds to reversing transitions
and turning final states into initial ones. 
Epi-mono factorization corresponds to taking the \emph{reachable} part of this automaton. Then, applying $\hat{G}$ reverses
transitions again, and turns initial states into final ones. Our abstract approach stops here; the original algorithm concludes
by taking the reachable part again, which ensures minimality. 


For a more detailed coalgebraic presentation of the full Brzozowski minimization algorithm in several concrete examples, see~\cite{BBHPRS14}. 
Another approach, based on duality theory, is presented in~\cite{BezhanishviliKP12}. The main idea there is similar, in the sense that contravariant adjunctions are lifted
to adjunctions between categories of coalgebras and algebras. However, it differs from the above development in that the adjunctions used in~\cite{BezhanishviliKP12} are assumed
to be dual equivalences, and the lifting of the duality is proved concretely for each example, rather than that a general condition is given. 
Another coalgebraic approach to minimization, based on partition refinement, is
in~\cite{AdamekBHKMS12}. It is mentioned in the conclusion that part of Brzozowski's algorithm appears
as an instance of the abstract construction introduced there, but the precise connection remains to be understood.

Notice that we only assume the mate of $\rho$ to be iso;
there are no requirements on $\alpha$. The mate of $\rho$ is iso for the logic 
from Example~\ref{ex:nda-bt}.
Thus, we can instantiate $\alpha$ to obtain observable deterministic automata from non-deterministic automata or even alternating automata
(by taking $T = \powf \powf$ and, for $\alpha$, the composition of $\alpha$ and $\beta$ from Example~\ref{ex:alt-aut}).
The logic $\theta$ from Example~\ref{ex:wa-correct} is covered as well, so one can treat Moore automata and weighted automata. 
However, the abstract construction of an observable automaton does not necessarily yield a concrete algorithm, 
as discussed for the case of weighted automata in~\cite{BBHPRS14}.



\begin{appendix}





\section{Details of Section~\ref{sec:examples}}\label{sec:details-examples}

In this section we expand on some of the examples considered in Section~\ref{sec:examples}. 

\paragraph{Example~\ref{ex:nda}.} 
We spell out the details of the semantics of the logic $\alpha \circledcirc \rho$ on an automaton $f \colon X \rightarrow \powf B X$:
\begin{align*}
	\tr(x)(\varepsilon) = \true
	  & \iff \alpha^\flat_{LA^*} (\powf\rho^\flat_{A^*} ((\powf B \tr)(f(x))))(*) = \true \\
	  & \iff \exists \varphi \in \powf\rho^\flat_{A^*} ((\powf B \tr)(f(x))) .\ \varphi(*) = \true \\ 
		& \iff \exists t \in f(x) .\ \rho^\flat_{A^*} (B \tr (t))(*) = \true \\
		& \iff * \in f(x) 
\end{align*}
where $\varepsilon$ is the empty word, 
and for all $a \in A$ and $w \in A^*$:
\begin{align*}
	\tr(x)(aw) = \true 
		& \iff \exists t \in f(x) .\ \rho^\flat_{A^*} (B \tr (t))(a,w) = \true \\
		& \iff \exists t \in f(x) .\ B \tr (t) = (a,\varphi) \wedge \varphi(w) = \true \\
		& \iff \exists y \in X .\ (a,y) \in f(x) \wedge \tr(y)(w) = \true
\end{align*}  

\paragraph{Example~\ref{ex:alt-aut}.} 
Spelling out $\tr$ yields $\tr(x)(\varepsilon)=o(x)$, and 
\begin{align*}
	\tr(x)(aw) = \true 
		& \iff \alpha_{A^*}^\flat (\powf \beta_{A^*}^\flat ((\powf \powf \tr)(f(x)(a))))(w)=\true \\
		& \iff \exists \varphi \in \powf \beta_{A^*}^\flat ((\powf \powf \tr)(f(x)(a))).\ \varphi(w)=\true \\
		& \iff \exists U \in(\powf \powf \tr) (f(x)(a)).\ \forall \varphi \in U.\ \varphi(w)=\true   \\
		& \iff \exists S \in f(x)(a).\ \forall y \in f(x)(a).\ \tr(y)(w) = \true 
\end{align*}

\paragraph{Example~\ref{ex:wta}.}
In order to understand the semantics $\tr$ of the forgetful logic on a tree automaton, we first compute the composite logic $\alpha^\flat_\Sigma \circ \M \rho^\flat \colon \M \Sigma \sem^- \Rightarrow S^\Sigma$:
\begin{align*}
	(\alpha^\flat_{\Sigma \Phi} \circ \M \rho^\flat_\Phi (\varphi))(\sigma(w_1, \ldots, w_n))
	& = \sum_{\psi \in \sem^{\Sigma \Phi}} (\M\rho^\flat_\Phi(\varphi))(\psi) \cdot \psi(\sigma(w_1, \ldots, w_n)) \\
	& = \sum_{\psi \in \sem^{\Sigma \Phi}} \sum_{\gamma \in {\rho^\flat_\Phi}^{-1}(\psi)}\varphi(\gamma) \cdot \psi(\sigma(w_1, \ldots, w_n)) \\
	& = \sum_{\varphi_1, \ldots, \varphi_n \in \sem^\Phi} \varphi(\sigma(\varphi_1, \ldots, \varphi_n)) \cdot \prod_{i = 1..n} \varphi_i(w_i)
\end{align*}
The next step is to instantiate this to $\Sigma^*\emptyset$ and precompose with $\M \Sigma \tr$:
\begin{align*}
	&(\alpha^\flat_{\Sigma \Sigma^*\emptyset} \circ \M \rho^\flat \circ \M \Sigma \tr(\psi))(\sigma(t_1, \ldots, t_n)) \\
		& = \sum_{\varphi_1, \ldots, \varphi_n \in \sem^{\Sigma^*\emptyset}} (\M \Sigma \tr (\psi))(\sigma(\varphi_1, \ldots, \varphi_n)) \cdot 
			\prod_{i=1..n} \varphi_i(t_i) \\
		& = \sum_{\varphi_1, \ldots, \varphi_n \in \sem^{\Sigma^*\emptyset}} \sum_{\substack{x_1 \in \tr^{-1}(\varphi_1) \\ \ldots \\ x_n \in \tr^{-1}(\varphi_n)}} 
		\psi(\sigma(x_1, \ldots, x_n))
		\cdot 
			\prod_{i=1..n} \varphi_i(t_i) \\
		& = \sum_{x_1, \ldots, x_n \in X} \psi(\sigma(x_1, \ldots, x_n)) \cdot \prod_{i = 1..n} \tr(x_i)(t_i)
\end{align*}
It follows that the diagram~\eqref{eq:wta-trace} commutes if and only if for all $\sigma(t_1, \ldots t_n)$
and all $x \in X$:
$$
\tr(x)(\sigma(t_1, \ldots, t_n)) = \sum_{x_1, \ldots, x_n \in X} f(x)(\sigma(x_1, \ldots, x_n)) \cdot \prod_{i = 1..n} \tr(x_i)(t_i)
$$
which is the semantics presented in Example~\ref{ex:wta}.

The map $\tr$ is computed by induction; this is done by turning a weighted tree automaton $f$
into the $\Sigma$-algebra $\hat{F}(f)$, which can be viewed as a \emph{deterministic weighted bottom-up tree automaton}. We 
spell out the details.
Given a coalgebra	$f \colon X \rightarrow \M \Sigma X$ the computed $\Sigma$-algebra looks as follows:
$$
\xymatrix{
	\Sigma (\sem^X) \ar[r]^-{\rho_X} 
		&\sem^{\Sigma X} \ar[r]^-{\alpha_{\Sigma X}} 
		&\sem^{\M \Sigma X} \ar[r]^-{\sem^{f}} 
		&\sem^X
}
$$
We have 
$$
	\alpha_X(\varphi)(\psi) = \sum_{x \in X} \varphi(x) \cdot \psi(x)
$$
and $\rho = \rho^\flat$:
\begin{align*}
	& (\sem^{\Sigma \eta_X} \circ \sem^{\rho^\flat_{\sem^X}} \circ \eta_{\Sigma \sem^X} (\sigma(\varphi_1, \ldots, \varphi_n)))(\tau(x_1, \ldots, x_m)) \\
		& = (\sem^{\rho^\flat_{\sem^X} \circ \Sigma \eta_X} \circ \lambda \psi. \psi(\sigma(\varphi_1, \ldots, \varphi_n)))(\tau(x_1, \ldots, x_m)) \\
		& = (\rho^\flat_{\sem^X} \circ \Sigma \eta_X(\tau(x_1, \ldots, x_m)))(\sigma(\varphi_1, \ldots, \varphi_n)) \\
		& = \rho^\flat_{\sem^X}(\tau(\lambda \varphi. \varphi(x_1), \ldots, \lambda \varphi.\varphi(x_m)))(\sigma(\varphi_1, \ldots, \varphi_n)) \\
		& = \begin{cases} \prod_{i = 1..n} \varphi_i(x_i) & \text{if } \sigma = \tau \\ 0 & \text{otherwise} \end{cases}
\end{align*}
The algebra is as follows:
\begin{align*}
	\hat{F}(X,f)(\sigma(\varphi_1, \ldots, \varphi_n))(x) 
	  & = (\sem^f \circ \alpha_{\Sigma X} \circ \rho_X(\sigma(\varphi_1, \ldots, \varphi_n)))(x) \\
		& = (\alpha_{\Sigma X} \circ \rho_X(\sigma(\varphi_1, \ldots, \varphi_n)))(f(x)) \\
		& = \sum_{t \in \Sigma X} f(x)(t) \cdot \rho_X(\sigma(\varphi_1, \ldots, \varphi_n))(t) \\
		& = \sum_{x_1, \ldots, x_n \in X} f(x)(\sigma(x_1, \ldots, x_n)) \cdot \prod_{i = 1..n} \varphi_i(x_i)
\end{align*}

\section{Proof of Lemma~\ref{lem:monadmap}}\label{app:sec5}
Recall the definitions:
\[
	\alpha^\dagger = \alpha^\flat F\circ T\iota, \qquad
	\alpha^\flat = G\epsilon\circ\alpha^\dagger G.
\]
First let us assume that $\alpha^\flat \colon TG\To G$ is a monad action, and check the axioms of monad morphisms for $\alpha^\dagger$. The unit axiom is the outer shape of:
\[\xymatrix{
	{\Id}\ar@{=>}[r]^-\iota\ar@{=>}[d]_\eta & GF\ar@{=}[rd]\ar@{=>}[d]_{\eta GF} \\
	T\ar@{=>}[r]_-{T\iota} & TGF\ar@{=>}[r]_{\alpha^\flat F} & GF
}\]
where the square commutes by naturality of $\eta$, and the triangle by a monad action axiom. The multiplication axiom is the outer shape of:
\[\xymatrix{
	TT\ar@{=>}[r]^-{TT\iota}\ar@{=>}[dd]_\mu & TTGF\ar@{=>}[r]^{T\alpha^\flat F}\ar@{=>}[dd]_{\mu GF} & TGF\ar@{=>}[r]^-{T\iota GF}\ar@{=}[rd] & TGFGF\ar@{=>}[r]^-{\alpha^\flat FGF}\ar@{=>}[d]^{TG\epsilon F} & GFGF\ar@{=>}[dd]^{G\epsilon F} \\
	& & & TGF\ar@{=>}[rd]^{\alpha^\flat F} \\
	T\ar@{=>}[r]_-{T\iota} & TGF\ar@{=>}[rrr]_{\alpha^\flat F} & & & GF
}\]
where everything commutes, from left to right: by naturality of $\mu$, by a monad action axiom, by~\eqref{eq:counit-unit} and by naturality of $\alpha^\flat$.

Now assume that $\alpha^\dagger \colon T\To GF$ is a monad morphism, and check the axioms of monad actions for $\alpha^\flat$. The unit axiom is the outer shape of:
\[\xymatrix{
& G\ar@{=>}[ld]_{\eta G}\ar@{=>}[d]^{\iota G}\ar@{=}[rd] \\
TG\ar@{=>}[r]_-{\alpha^\dagger G} & GFG\ar@{=>}[r]_-{G\epsilon} & G
}\]
where the left triangle is the unit axiom of a monad morphism, and the right triangle is~\eqref{eq:counit-unit}. The multiplication axiom is the outer shape of:
\[\xymatrix{
TTG\ar@{=>}[r]^-{T\alpha^\dagger G}\ar@{=>}[dd]_{\mu G} & TGFG\ar@{=>}[r]^{TG\epsilon}\ar@{=>}[d]_{\alpha^\dagger GFG} & TG\ar@{=>}[d]^{\alpha^\dagger G} \\
& GFGFG\ar@{=>}[r]_-{GFG\epsilon}\ar@{=>}[d]_{G\epsilon FG} & GFG\ar@{=>}[d]^{G\epsilon} \\
TG\ar@{=>}[r]_-{\alpha^\dagger G} & GFG\ar@{=>}[r]_{G\epsilon} & G
}\]
where the leftmost shape is the multiplication axiom of a monad morphism, and the two smaller squares commute by naturality of $\alpha^\dagger$ and $\epsilon$.

\section{Details of examples in Section~\ref{sec:determinization}}\label{app:sec6}

\paragraph{Example~\ref{ex:wa-correct}.}
The condition from Theorem~\ref{thm:det-translate} is commutativity of the following diagram:
$$
\xymatrix{
	\M(A \times \sem^- + 1) \ar@{=>}[r]^-{\M \rho^\flat} \ar@{=>}[d]^{\kappa_{\sem^-}} 
		& \M (\sem^{A \times - + 1}) \ar@{=>}[r]^-{\alpha^\flat_{A \times - + 1}} 
		& \sem^{A \times - + 1} \ar@{=}[d] \\
	\sem \times (\M \sem^-)^A \ar@{=>}[r]^-{\id \times (\alpha^\flat)^A} 
		& \sem \times (\sem^-)^A \ar@{=>}[r]^-{\theta^\flat} 
		& \sem^{A \times - + 1}
}
$$
Indeed, we have
$$
 K\alpha^\flat_\Phi \circ \kappa_{\sem^\Phi}(\varphi) = 
		(\varphi(*), \lambda a. \alpha^\flat(\lambda \psi. \varphi(a,\psi))) 
		= (\varphi(*), \lambda a. \lambda w. \sum_{\psi \in \sem^\Phi} \varphi(a, \psi) \cdot \psi(w))
$$
and thus
\begin{align*}
(\theta^\flat_\Phi \circ K\alpha^\flat_\Phi \circ \kappa_{\sem^\Phi}(\varphi))(*) &= \varphi(*)\\
	(\theta^\flat_\Phi \circ K\alpha^\flat_\Phi \circ \kappa_{\sem^\Phi}(\varphi))(a,w) 
	& = \sum_{\psi \in \sem^\Phi} \varphi(a, \psi) \cdot \psi(w)
\end{align*}
which coincides with $\alpha^\flat_{A \times \Phi +1} \circ \M \delta^\flat_\Phi$ as computed (in a more general setting) in Appendix~\ref{sec:details-examples}.

\paragraph{Example~\ref{ex:nda-correct}.}
We treat the determinization $\langle \tau^o, \tau^t \rangle$ described in the example. The relevant condition of Theorem~\ref{thm:det-translate}
instantiates to commutativity of:
$$
\xymatrix{
	\powf(2 \times (2^-)^A) \ar@{=>}[r]^-{\powf \rho^\flat} \ar@{=>}[d]^-{\langle \tau^o_{2^-}, \tau^t_{2^-} \rangle} 
		& \powf(2^{A \times - + 1}) \ar@{=>}[r]^-{\beta^\flat L} 
		& 2^{A \times - + 1} \ar@{=}[d] \\
	2 \times (\powf(2^-))^A \ar@{=>}[r]^-{\id \times (\beta^\flat)^A} 
		& 2 \times (2^-)^A \ar@{=>}[r]^-{\rho^\flat}
		& 2^{A \times - + 1}
}
$$
We have, for any set $\Phi$:
\begin{align*}
	(\beta^\flat_{L\Phi} \circ \powf \rho^\flat_{\Phi})(S)(*) = \true 
		& \iff \forall \varphi \in (\powf\rho^\flat_\Phi)(S) . \varphi(*)=\true \\
		& \iff \forall (o,t) \in S . o = \true \\
		& \iff \tau^o_{2^\Phi}(S) = \true \\
		& \iff (\rho^\flat_\Phi \circ B\beta^\flat_\Phi \circ \langle \tau^o_{2^\Phi}, \tau^t_{2^\Phi} \rangle)(S)(*) = \true
\end{align*}
and for any $a \in A$, $w \in \Phi$:
\begin{align*}
  (\beta^\flat_{L\Phi} \circ \powf \rho^\flat_{\Phi})(S)(a,w) = \true 
  	& \iff \forall \varphi \in (\powf \rho^\flat_\Phi)(S). \varphi(a,w) = \true \\
  	& \iff \forall (o, t) \in S. t(a)(w) = \true \\
  	& \iff \forall \varphi \in \tau_{2^\Phi}^t(S)(a) . \varphi(w) = \true \\
  	& \iff \beta^\flat_\Phi (\tau^t_{2^\Phi}(S)(a))(w) = \true \\
  	& \iff ((\beta^\flat_\Phi)^A \circ \tau^t_{2^\Phi})(S)(a)(w) = \true \\
  	& \iff (\rho^\flat_\Phi \circ B\beta^\flat_\Phi \circ \langle \tau^o_{2^\Phi}, \tau^t_{2^\Phi} \rangle)(S)(a,w) = \true
\end{align*}
which proves commutativity of the diagram.

\paragraph{Example~\ref{ex:alt-correct}.}
First, we check that $\chi$ defined by~\eqref{eq:goodChi} is a natural transformation. To this end, for any function $f \colon X\to Y$, and a family $S\in \powf\powf X$, we need to check that
\[
	\chi_Y((\powf\powf f)(S)) = (\powf\powf f)(\chi_X(S)).
\]
For any $W\subseteq Y$, calculate:
\begin{align}
	 W\in\chi_Y((\powf\powf f)(S)) &\iff W\subseteq\bigcup(\powf\powf f)(S) \land \forall T\in(\powf\powf f)(S).\ W\cap T\neq\emptyset  \nonumber \\
	&\iff W\subseteq\powf f\left(\bigcup S\right)\land \forall U\in S.\ W\cap (\powf f)(U) \neq \emptyset \label{eq:B1}\\
	W\in(\powf\powf f)(\chi_X(S)) &\iff \exists V\in\chi_X(S).\ W=(\powf f)(V)  \nonumber \\
	&\iff \exists V\subseteq \bigcup S.\ (W=(\powf f)(V) \land \forall U\in S.\ V\cap U\neq\emptyset)\label{eq:B2}
\end{align}
For the second equivalence in~\eqref{eq:B1}, notice that
\[
	\bigcup(\powf\powf f)(S) = \powf f\left(\bigcup S\right)
\]
for any $f \colon X\to Y$ and $S\subseteq\powf X$, since direct images preserve unions.

We need to show that~\eqref{eq:B1} and~\eqref{eq:B2} are equivalent for every $W$. 

For the implication from~\eqref{eq:B1} to ~\eqref{eq:B2}, put 
\[
   V=\overleftarrow{f}(W) \cap \bigcup S.
\]
Then $W=(\powf f)(V)$. Indeed, for the left-to-right containment, pick any $y\in W$. Then, since $W\subseteq\powf f\left(\bigcup S\right)$, there is some $x\in \bigcup S$ such that $f(x)=y$. Obviously then $x\in\overleftarrow{f}(W)$, hence $x\in V$ and $y\in\powf f (V)$. The right-to-left containment is equivalent to $\overleftarrow{f}(W)\supseteq V$, which follows directly from the definition of $V$.

Moreover, for any $U\in S$, by~\eqref{eq:B1} there exists some $y\in W\cap (\powf f)(U)$, so there is some $x\in U$ such that $f(x)=y$. Then $x\in \overleftarrow{f}(W)$. Obviously $x\in\bigcup S$ as well, so $x\in V$ and $V\cap U\neq\emptyset$.

For the implication from~\eqref{eq:B2} to ~\eqref{eq:B1}, take a $V\subseteq \bigcup S$ that exists by~\eqref{eq:B2}, and calculate:
\[
	W = (\powf f)(V) \subseteq (\powf f)\left(\bigcup S\right).
\]
Furthermore, for any $U\in S$, by~\eqref{eq:B2} we have $V\cap U\neq\emptyset$. Then calculate:
\[
	W\cap(\powf f)(U) = (\powf f)(V)\cap (\powf f)(U) \supseteq (\pow f)(V\cap U)\neq\emptyset.
\]
We thus conclude that $\chi \colon \powf\powf\To\powf\powf$ is a natural transformation.

\bigskip

We now prove that the diagram~\eqref{eq:chialphabeta} commutes. Recall that, for any set $X$ and any $S\in\powf\powf 2^X$, $U\in\powf 2^X$ and $x\in X$:
\[\begin{array}{rcl}
\chi_{2^X}(S) &=& \{V\subseteq \bigcup S \mid \forall U\in S.\ V\cap U\neq\emptyset\} \\[1ex]
\alpha^\flat_{X}(U)(x)=\true &\iff& \exists\phi\in S.\ \phi(x)=\true \\[1ex]
\beta^\flat_{X}(U)(x)=\true &\iff& \forall\phi\in S.\ \phi(x)=\true
\end{array}\]
Then calculate, for any $S\in\powf\powf 2^X$ and $x\in X$:
\begin{align}
\beta^\flat_X(\powf\alpha^\flat_X(S))(x)=\true
	&\iff \forall\phi\in\powf\alpha^\flat_X(S).\ \phi(x)=\true \nonumber \\
	&\iff \forall U\in S.\ \alpha^\flat_X(U)(x)=\true \nonumber \\
	&\iff \forall U\in S.\ \exists\phi\in U. \phi(x)=\true \label{eq:B3}\\
\alpha^\flat_X(\powf\beta^\flat_X(\chi_{2^X}(S)))(x) = \true	
	&\iff \exists\phi\in\powf\beta^\flat_X(\chi_{2^X}(S)).\ \phi(x)=\true \nonumber \\
	&\iff \exists V\in \chi_{2^X}(S).\ \beta^\flat_X(V)(x)=\true \nonumber \\
	&\iff \exists V\in \chi_{2^X}(S).\ \forall\phi\in V. \phi(x)=\true \nonumber \\
	&\iff \exists V\subseteq \bigcup S.(\forall U\in S.\ U\cap V\neq\emptyset)\land (\forall\phi\in V. \phi(x)=\true)\label{eq:B4}
\end{align}
Conditions~\eqref{eq:B3} and~\eqref{eq:B4} are easily equivalent, therefore~\eqref{eq:chialphabeta} commutes.

\bigskip

Finally, the condition of Theorem~\ref{thm:det-translate} in this case is that the following commutes:
	$$
		\xymatrix{
			\powf B \powf G \ar@{=>}[r]^-{\powf B\alpha^\flat} \ar@{=>}[d]^-{\tau {\powf G}} 
				& \powf BG \ar@{=>}[r]^-{\powf \rho^\flat} \ar@{=>}[d]^-{\tau G}
				& \powf GL \ar@{=>}[r]^-{\beta^\flat L}
				& GL \\
			B\powf\powf G \ar@{=>}[r]^-{B\powf\alpha^\flat} \ar@{=>}[d]^-{B\chi G}
				& B\powf G \ar@{=>}[r]^-{B\beta^\flat} 
				& BG \ar@{=>}[ur]^-{\rho^\flat} & \\
			B\powf\powf G \ar@{=>}[r]^-{B\powf \beta^\flat} 
				& B\powf G \ar@{=>}[ru]^-{B\alpha^\flat}
				& &			
		}
	$$
	The square commutes by naturality of $\alpha^\flat$, and the upper right shape is the diagram for proving correctness of 
	the determinization procedure $\tau$ considered in Example~\ref{ex:nda-correct}.
	The lower shape is~\eqref{eq:chialphabeta} mapped by $B$.

\end{appendix}

\end{document}